\pdfoutput=1
\documentclass[11pt]{article}
\usepackage{amsmath,amssymb,amsthm}
\usepackage{bm}
\usepackage{amscd}
\usepackage{graphicx}
\usepackage{cancel}
\usepackage{url}
\usepackage{tikz}
\usetikzlibrary{cd}
\usepackage{ytableau}
\newtheorem{thm}{Theorem}
\newtheorem{prop}[thm]{Proposition}
\newtheorem{lemma}[thm]{Lemma}

\newtheorem{rem}[thm]{Remark}
\newtheorem{example}[thm]{Example}
\newcommand{\arc}{\stackrel{\frown}{1\,0}}
\newcommand{\arczeroone}{\stackrel{\frown}{0\,\,1}}
\makeatletter
\@addtoreset{equation}{section}

\makeatother
\title{Linearization of the box-ball system: an elementary approach}
\author{{\sc Saburo Kakei}\\
\textit{\small Department of Mathematics, Rikkyo University,
Toshima-ku, Tokyo 171-8501, Japan}
\\[2mm]
{\sc Jonathan J.C. Nimmo}\\
\textit{\small School of Mathematics and Statistics, University of Glasgow, 
Glasgow G12 8QW, UK}
\\[2mm]
{\sc Satoshi Tsujimoto}\\
\textit{\small Department of Applied Mathematics and Physics, 
Graduate School of Informatics,}\\
\textit{\small Kyoto University, Sakyo-ku, Kyoto 606-8501, Japan}
\\[2mm]
{\sc and}
\\[2mm]
{\sc Ralph Willox}\\
\textit{\small Graduate School of Mathematical Sciences, 
the University of Tokyo,
Meguro-ku, 153-8914 Tokyo, Japan}
}
\date{}

\begin{document}
\maketitle
\begin{abstract}
Kuniba, Okado, Takagi and Yamada have found that the time-evolution of 
the Takahashi-Satsuma box-ball system can be linearized by considering
rigged configurations associated with states of the box-ball system. 
We introduce a simple way to understand the rigged configuration of
$\mathfrak{sl}_2$-type, and give an elementary proof of the linearization
property. 
Our approach can be applied to a box-ball system with finite carrier, 
which is related to a discrete modified KdV equation, and also to
the combinatorial $R$-matrix of $A_1^{(1)}$-type. 
We also discuss combinatorial statistics 
and related fermionic formulas 
associated with the states of the box-ball systems. 
A fermionic-type formula we obtain for the finite carrier case 
seems to be new.
\end{abstract}


\section{Introduction}
The box-ball system (BBS for short) was introduced in 1990 as a 
cellular automaton model that exhibits solitonic behaviour \cite{TS}.
Since then it has been studied from various perspectives 
such as ultra-discretisation of discrete soliton equations 
\cite{TTMS,TM,TsujimotoHirota,KNW2009,KNW2010,GNN}, 
representation theory of quantum groups \cite{HHIKTT,IKT,TakagiReview}, 
and combinatorics \cite{TTS,Ariki,Fukuda,FOY}. In particular, it is
known to be related to the ultra-discrete limit of the discrete KdV
equation \cite{TTMS,TsujimotoHirota,KNW2009}, a link which allowed for the
obtention of its $N$-soliton solution in \cite{TTMS,MIT2} and for the solution of
its general initial value problem by means of IST techniques, similar to those for the continuous KdV equation, in
\cite{RWudKdV1,RWudKdV2}.

Kuniba, Okado, Takagi, and Yamada found that the time-evolution of the
BBS can be linearized by considering rigged configurations 
associated with states of the BBS \cite{KOTY,KOSTY,Takagi,TakagiReview}. 
Originally, rigged configurations were introduced as combinatorial 
objects that label the solutions to the Bethe ansatz equations for 
integrable spin chains \cite{KKR,SchillingReview} and later they were investigated from the viewpoint of 
Kashiwara crystals \cite{HKOTT,OkadoRevrew,Schilling:X=M}.
The linearization property for the BBS was conjectured in \cite{KOTY} and 
proved in \cite{Takagi,KOSTY,Sakamoto1}. The original proof in \cite{Takagi} is 
formulated in terms of integer-valued two-row matrices that correspond to 
states of the BBS. 
A representation-theoretical proof has been presented in \cite{KOSTY,Sakamoto1}.
In the works \cite{KTT,Takagi,TakagiReview}, the term ``inverse
scattering transform (IST)'' is used to indicate that the time-evolution 
is linearized in terms of riggings, 
which is different from the ``IST'' in \cite{RWudKdV1,RWudKdV2}. 

The linearization property is useful in considering the 
initial value problem for the BBS with periodic boundary condition
\cite{KTT}. 
Mada, Idzumi and Tokihiro developed  
another approach for the same initial value problem based on 
``10-eliminations'' \cite{MIT}. 
The relationship between these two approaches, 
rigged configurations and 10-eliminations, 
has been discussed in great detail in \cite{KS} in crystal-theoretic terms.

In this paper we shall give a simple and elementary proof of the
linearization property, based on the correspondence between 
``$10$-eliminations'' and ``$01$-eliminations''. Our construction has a
more visual flavour than previous approaches and 
can be easily extended to the time-evolution of a BBS with finite
carrier \cite{TM}, for which we shall also establish its linearization
in elementary terms. 

Throughout the paper, we shall use the following notation:
\begin{itemize}
\item Semi-infinite binary sequences:
$\bold{u}=(u_0,u_1,u_2,\ldots)$, 
$u_j\in\{0,1\}$ ($j=0,1,2,\ldots$).\\
For later convenience we impose the restriction that $u_0=0$. 
\item The $j$th component of $\bold{u}=(u_0,u_1,u_2,\ldots)$: 
$(\bold{u})_j = u_j$.
\item The truth function $\chi$: $\chi(A)=1$ if $A$ is true, 
and $\chi(A)=0$ otherwise.
\item The number of ``balls'' in $\bold{u}$: 
$\mathcal{N}(\bold{u})=\sum_{j=0}^\infty
\chi\left((\bold{u})_j=1\right)$.
\end{itemize}
We also use some terminology and notation that is standard in combinatorics 
\cite{MacMahon,Mansour,Krattenthaler}:
\begin{itemize}
\item The descent number of $\bold{u}$: 
\begin{equation}
\mathrm{des}(\bold{u})=
\sum_{j=0}^{\infty}\chi\left(
(\bold{u})_j>(\bold{u})_{j+1}
\right), 
\end{equation}
which is the number of times 
the ``$10$'' pattern appears in $\bold{u}$.
\item The descent sequence for $\bold{u}$:
\begin{equation}
\mathrm{Des}(\bold{u})=\left\{ 
d_j\in\mathbb{Z}_{\geq 0}
\,\big|\,(\bold{u})_{d_j}>(\bold{u})_{d_j+1},\;
d_1>\cdots >
d_{\mathrm{des(\bold{u})}}
\right\}, 
\end{equation}
obtained from the positions of the descends in $\bold{u}$ 
(Figure \ref{fig:def_d_j}).
\item The ascent number of $\bold{u}$: 
\begin{equation}
\mathrm{asc}(\bold{u})=
\sum_{j=0}^{\infty}\chi\left(
(\bold{u})_j<(\bold{u})_{j+1}
\right), 
\end{equation}
which is the number of times the ``$01$'' pattern 
appears in $\bold{u}$.
\item The ascent sequence for $\bold{u}$: 
\begin{equation}
\mathrm{Asc}(\bold{u})=\left\{
a_j\in\mathbb{Z}_{\geq 0}\,\big|\,
(\bold{u})_{a_j}<(\bold{u})_{a_j+1},\;
a_1>\cdots >a_{\mathrm{asc}(\bold{u})}
\right\}, 
\end{equation}
obtained from the positions of the ascends in $\bold{u}$ 
(Figure \ref{fig:def_a_j}). 
\end{itemize}

\begin{figure}[htbp]
\begin{center}
\[
\begin{array}{ccccccccccccccccc}
j&:&0 & \cdots & d_N & \cdots & \cdots & d_{N-1} & \cdots & \cdots &
 \cdots & d_2 & \cdots & \cdots & d_1 & \cdots & \cdots\\
(\bold{u})_j&:&0 & \cdots & 1 & 0& \cdots & 1 & 0 & \cdots & \cdots & 1 & 0 & \cdots &
 1 & 0 & \cdots
\end{array}
\]
\caption{$\mathrm{Des}(\bold{u})=\{d_1>d_2>\cdots>d_N\}$, 
$N=\mathrm{des}(\bold{u})$}
\label{fig:def_d_j}
\end{center}
\end{figure}
\begin{figure}[htbp]
\begin{center}
\[
\begin{array}{ccccccccccccccccc}
j&:&0 & \cdots & a_N & \cdots & \cdots & a_{N-1} & \cdots & \cdots &
 \cdots & a_2 & \cdots & \cdots & a_1 & \cdots & \cdots\\
(\bold{u})_j&:&0 & \cdots & 0 & 1& \cdots & 0 & 1 & \cdots & \cdots & 0 & 1 & \cdots &
 0 & 1 & \cdots
\end{array}
\]
\caption{$\mathrm{Asc}(\bold{u})=\{a_1>a_2>\cdots>a_N\}$, 
$N=\mathrm{asc}(\bold{u})$}
\label{fig:def_a_j}
\end{center}
\end{figure}

This article is organized as follows;
In Section \ref{sec:TS-BBA}, using the above notation, 
we present the necessary background on the box-ball system needed to
prove the linearization property. In particular, we introduce the notion
of ``01-elimination with rigging'', which is actually equivalent to that
of a rigged configuration. 
In Section \ref{sec:finite}, we extend our approach
to the BBS with a carrier of finite size. 
In Section \ref{sec:statistics}, we consider 
combinatorial statistics associated with the BBS 
and its relation to fermionic formulas. 
We derive a similar fermionic-type formula for the 
BBS with finite carrier. 

\section{Takahashi-Satsuma box-ball system}
\label{sec:TS-BBA}
\subsection{Time-evolution, 10-eliminations and 01-eliminations}
\label{subsec:TimeEvolution_TakahashiSatsuma}
The time-evolution rule of the BBS can be formulated as an
operation on binary sequences \cite{TS,YYT}. 
For later convenience, we only consider binary sequences 
start with $u_0=0$.
Denoting by $\mathcal{U}$ the set of possible BBS states
\begin{equation}
\mathcal{U}=\left\{
\bold{u} \,\big|\, u_0=0,\, \mathcal{N}(\bold{u})\mbox{ is finite} \right\},
\end{equation}
we have that
\begin{equation}
\mathrm{des}(\bold{u})=\mathrm{asc}(\bold{u})
\mbox{ for all }\bold{u}\in\mathcal{U}, 
\label{N10=N01}
\end{equation}
which is the so-called ``soliton number'' of $\bold{u}$,
and we define the subset $\mathcal{U}_N\subset\mathcal{U}$ 
(the so-called $N$-soliton sector) as 
\begin{equation}
\mathcal{U}_N=\left\{\bold{u}\in\mathcal{U}\,\big|\,
\mathrm{des}(\bold{u})=\mathrm{asc}(\bold{u})=N\right\}.
\end{equation}
A decreasing sequence of integers that satisfies the interlacing condition
\begin{equation}
\label{InterlacingProperty}
d_1>a_1>\cdots>d_N>a_N\geq 0
\end{equation}
uniquely parametrises a 
semi-infinite sequence $\bold{u}\in\mathcal{U}_N$, 
which we denote by $\bold{u}(a_1,\ldots,a_N;d_1,\ldots,d_N)$
(cf. Figure \ref{fig:u(a_1,...,a_N;d_1,...,d_N)}). 

\begin{figure}[htbp]
\[
\begin{array}{ccc@{\ }c@{\ }cc@{\ }c@{\ }ccccc@{\ }c@{\ }cc@{\ }c@{\ }ccc}
j&:&0 & \cdots & a_N & \cdots & \cdots & d_N & \cdots & \cdots 
&d_2 & \cdots & \cdots & a_1 & \cdots & \cdots & d_1 & \cdots&\cdots\\
(\bold{u})_j&:&0 & \cdots & 0 & 1 & \cdots & 1 & 0 & \cdots &1&0& \cdots & 0 &
1 & \cdots & 1 & 0 & \cdots
\end{array}
\]
\caption{$\bold{u}(a_1,\ldots,a_N;d_1,\ldots,d_N)$}
\label{fig:u(a_1,...,a_N;d_1,...,d_N)}
\end{figure}

The time-evolution $T:\mathcal{U}\to\mathcal{U}$ can be 
described by drawing ``10-arc lines'' \cite{YYT} according to the
following simple rules
(Figure \ref{fig:time-evolution}): 
\begin{itemize}
\item[i)] For $\bold{u}\in\mathcal{U}$, 
connect all $10$ pairs with arc lines (``1st 10-arc lines'').
\item[ii)] 
Disregarding the 1s and 0s in the already connected $10$ pairs, connect all the 
remaining
$10$ pairs with arc lines (``2nd 10-arc lines''). 
\item[iii)] 
Repeat the above procedure until all the 1s are connected to 0s.
\item[iv)] 
Define $T(\bold{u})$ as the state obtained by 
exchanging the 1s and 0s in every connected $10$ pair.
\end{itemize}

\begin{figure}[htbp]
\begin{center}
\begin{tikzpicture}
\path (1.15,2) node {$\bold{u}$ \ \ $=$};
\path (2,2) node {$0$};
\path (2.5,2) node {$0$};
\path (3,2) node {$1$};
\path (3.5,2) node {$1$};
\path (4,2) node {$1$};
\path (4.5,2) node {$0$};
\path (5,2) node {$0$};
\path (5.5,2) node {$1$};
\path (6,2) node {$0$};
\path (6.5,2) node {$0$};
\path (7,2) node {$0$};
\path (7.5,2) node {$0$};
\path (8,2) node {$\cdots$};
\draw (4,2.2) to [out=60,in=120] (4.5,2.2);
\draw (5.5,2.2) to [out=60,in=120] (6,2.2);
\draw (3.5,2.2) to [out=60,in=120] (5,2.2);
\draw (3,2.2) to [out=60,in=120] (6.5,2.2);
\path (0,0) node {$\to$};
\path (1,0) node {$T(\bold{u})=$};
\path (2,0) node {$0$};
\path (2.5,0) node {$0$};
\path (3,0) node {$0$};
\path (3.5,0) node {$0$};
\path (4,0) node {$0$};
\path (4.5,0) node {$1$};
\path (5,0) node {$1$};
\path (5.5,0) node {$0$};
\path (6,0) node {$1$};
\path (6.5,0) node {$1$};
\path (7,0) node {$0$};
\path (7.5,0) node {$0$};
\path (8,0) node {$\cdots$};
\draw (4,0.2) to [out=60,in=120] (4.5,0.2);
\draw (5.5,0.2) to [out=60,in=120] (6,0.2);
\draw (3.5,0.2) to [out=60,in=120] (5,0.2);
\draw (3,0.2) to [out=60,in=120] (6.5,0.2);
\end{tikzpicture}
\caption{Example of the BBS time-evolution}
\label{fig:time-evolution}
\end{center}
\end{figure}
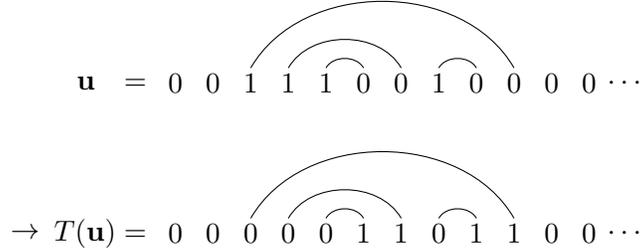

One can also draw ``01-arc lines'' for $\bold{u}$ in the same
manner. 
The following lemma is obvious from the definition of
$T(\bold{u})$, as discussed in \cite{YYT}. 
However, it will turn out to play a crucial role in our approach.
\begin{lemma}
\label{lemma:10-arcs_01-arcs}
The 10-arc lines for $\bold{u}\in\mathcal{U}$ coincide with 
the 01-arc lines for $T(\bold{u})$.
\end{lemma}

Next we introduce the  ``10-elimination'' procedure. 
We first prepare a map $\phi_{a}:\mathcal{U}\to\mathcal{U}$, where $a$
is a non-negative integer:
\begin{equation}
\left(\phi_a(\bold{u})\right)_j = 
\begin{cases}
(\bold{u})_j & (j<a),\\
(\bold{u})_{j+2} & (j\geq a).
\end{cases}
\end{equation}
Denote by $\Phi_{10}(\bold{u})$ the 10-eliminated sequence of
$\bold{u}$: 
\begin{equation}
\Phi_{10}(\bold{u}) = \left(
\phi_{d_N}\circ\cdots\circ\phi_{d_2}\circ\phi_{d_1}
\right)(\bold{u}), \quad
\{d_1>\cdots > d_N\}=\mathrm{Des}(\bold{u}).
\end{equation}
The 01-eliminated sequence $\Phi_{01}(\bold{u})$ 
can be described in the same fashion: 
\begin{equation}
\Phi_{01}(\bold{u}) = \left(
\phi_{a_N}\circ\cdots\circ\phi_{a_2}\circ\phi_{a_1}
\right)(\bold{u}), \quad
\{a_1>\cdots > a_N\}=\mathrm{Asc}(\bold{u}).
\end{equation}
We remark that $\Phi_{10}$ can act on $\mathcal{U}$ repeatedly. 
In the case of the 01-elimination, $\Phi_{01}$ 
can act on $\bold{u}\in\mathcal{U}$ at least once but not always twice 
since it might happen that $\Phi_{01}(\bold{u})\notin \mathcal{U}$
(e.g. $\bold{u}=011000\cdots$). 

Define $\mathcal{U}^{(n)}$ ($n=1,2,\ldots$) recursively by 
\begin{equation}
\mathcal{U}^{(1)}=\mathcal{U}, \qquad 
\mathcal{U}^{(n)}=\Phi_{01}^{-1}\left(\mathcal{U}^{(n-1)}\right)
\quad (n=2,3,\ldots).
\end{equation}
We also define $\mathcal{U}^+$ (the set of ``lattice words'') as 
\begin{equation}
\mathcal{U}^+=\left\{
\bold{u}\in\mathcal{U}\,\Bigg|\,
\sum_{j=0}^{k}\chi\left((\bold{u})_j=0\right)
\geq \sum_{j=0}^{k}\chi\left((\bold{u})_j=1\right)
\mbox{ for }k=1,2,\ldots
\right\}.\end{equation}
For $\bold{u}\in\mathcal{U}^{(n)}$, 
$\Phi_{01}$ can act at least $n$-times, 
and the following relation holds:
\begin{equation}
\mathcal{U}=\mathcal{U}^{(1)}\supset\mathcal{U}^{(2)}
\supset\cdots\supset\mathcal{U}^{(n)}
\supset\cdots\supset\mathcal{U}^{+}.
\end{equation}

Define the forward-shift operator $\Lambda$ on $\mathcal{U}$ as 
\begin{equation}
\left(\Lambda(\bold{u})\right)_n = 
\begin{cases}
0 & (n=0),\\
(\bold{u})_{n-1} & (n=1,2,\ldots).
\end{cases}
\end{equation}
It follows that
\begin{equation}
\Lambda\circ T=T\circ\Lambda, \quad
\mathrm{asc}\circ \Lambda = \mathrm{asc}, \quad 
\mathrm{des}\circ \Lambda = \mathrm{des},
\label{N10_Lambda=N10}
\end{equation}
and 
\begin{equation}
\Phi_{10}=\Lambda\circ\Phi_{01}=\Phi_{01}\circ\Lambda.
\label{Phi10=LambdaPhi01=Phi01Lambda}
\end{equation}
Note that the transformations $T$, $\Lambda$, $\Phi_{01}$, and
$\Phi_{10}$ on $\mathcal{U}$ can be restricted to $\mathcal{U}^+$. 

\begin{lemma}
\label{lemma:Lambda_T_Phi10=Phi10_T}
For all $\bold{u}\in\mathcal{U}^+$, 
$\left(\Lambda^k\circ T\circ\Phi_{01}^k\right)(\bold{u})
=\left(\Phi_{01}^k\circ T\right)(\bold{u})$ ($k=0,1,2,\ldots$).
\end{lemma}
\begin{proof}
{}From \eqref{N10_Lambda=N10} and 
\eqref{Phi10=LambdaPhi01=Phi01Lambda}, we have
\begin{equation}
\Lambda^k\circ T\circ\Phi_{01}^k
= T\circ \left(\Lambda\circ\Phi_{01}\right)^k
= T\circ \Phi_{10}^k.
\label{Lambda^kTPhi01^k=T(LambdaPhi01)^k}
\end{equation}
This relation, together with the following direct consequence of 
Lemma \ref{lemma:10-arcs_01-arcs},
\begin{equation}
\label{TPhi=PhiT}
T\circ\Phi_{10}= \Phi_{01}\circ T,
\end{equation}
then yields the desired result.
\end{proof}

For $\bold{u}\in\mathcal{U}^+$, define $\ell$ as the 
minimal integer such that $\mathcal{N}\left(
\Phi_{01}^{\ell}(\bold{u})\right)= 0$. 
Then we define 
$\lambda_i(\bold{u})$ ($i=1,2,\ldots, \ell$) as 
the number of $i$th 01-arc lines associated with $\bold{u}$, i.e.,
\begin{equation}
\lambda_i(\bold{u}) = \left(\mathrm{asc}\circ\Phi_{01}^{i-1}\right)(\bold{u}).
\label{def:lambda_i}
\end{equation}
The integers $\lambda_1(\bold{u})$, 
$\ldots$, $\lambda_\ell(\bold{u})$ clearly
satisfy $\lambda_1(\bold{u})\geq
\cdots\geq\lambda_\ell(\bold{u})\geq 0$ 
and thus $\lambda(\bold{u})=\left(\lambda_1(\bold{u}),
\ldots\lambda_\ell(\bold{u})\right)$ is a
partition.
\begin{thm}[cf. \cite{TTS,Ariki,YYT}]
\label{thm:invariants0}
For $\bold{u}\in\mathcal{U}^+$, define a partition $\lambda(\bold{u})$ as above. 
Then $\lambda\left(\bold{u}\right)$ is invariant 
under the time-evolution $T$, i.e., 
$\lambda\left(T(\bold{u})\right)=\lambda(\bold{u})$. 
\end{thm}
\begin{proof}
Define $\overline{\lambda}_i(\bold{u})$ ($i=1,2,\ldots,\ell$) as
\begin{equation}
\overline{\lambda}_i(\bold{u}) = \left(\mathrm{asc}\circ\Phi_{01}^{i-1}
\circ T\right)(\bold{u}).
\end{equation}
{}From Lemma \ref{lemma:10-arcs_01-arcs} and relation \eqref{N10=N01} we have 
\begin{equation}
\label{N10=N10_T}
\mathrm{asc}\circ T = \mathrm{des} = \mathrm{asc}.
\end{equation}
Using Lemma \ref{lemma:Lambda_T_Phi10=Phi10_T} and 
the relations \eqref{N10_Lambda=N10}, \eqref{N10=N10_T}, 
we obtain
\begin{equation}
\overline{\lambda}_i(\bold{u})
= \left(\mathrm{asc}\circ\Lambda^{i-1}\circ
   T\circ\Phi_{01}^{i-1}\right)(\bold{u})
= \lambda_i(\bold{u}).
\end{equation}
Thus we have $\overline{\lambda}_i(\bold{u}) = \lambda_i(\bold{u})$ 
for all $i$. 
\end{proof}

We define $\mu$ as the partition conjugate to $\lambda$. 
The partition $\mu$ is of course also invariant under time-evolution
and it coincides, in fact, with the invariants of the BBS 
that were introduced in \cite{TTS} and discussed in \cite{Ariki}.
The proofs of the invariance property in \cite{TTS,Ariki} are based on
the Dyck language.
It is easily seen that the conjugate partition $\mu=
(\mu_1,\mu_2,\ldots,\mu_{\lambda_1})$ gives the lengths of the solitons
that arise asymptotically from the state $\bold{u}$ under the BBS
evolution. 
We shall therefore refer to the partition $\mu$ as  the 
``asymptotic soliton contents'' of $\bold{u}\in\mathcal{U}$. 
\begin{thm}[asymptotic soliton contents \cite{TTS}] 
\label{thm:SolitonContents_Tinfinity}
For $\bold{u}\in\mathcal{U}_N$, let $\mu=
(\mu_1,\mu_2,\ldots,\mu_{\lambda_1})$ be as above, and
denote $\mathrm{Asc}\left(T^k(\bold{u})\right)$ 
by $\left\{a_j(k)\right\}_{j=1,\ldots,N}$ and 
$\mathrm{Des}\left(T^k(\bold{u})\right)$ by 
$\left\{d_j(k)\right\}_{j=1,\ldots,N}$.
Then there exists an integer $K$ such that
\begin{equation}
\begin{aligned}
 d_j(k)-a_j(k) &= \mu_j && (j=1,\ldots,N), \\
a_j(k)-d_{j+1}(k) &\geq \mu_{j+1} && (j=1,\ldots,N-1)
\end{aligned}
\label{d-a=mu,a-d>=mu}
\end{equation}
for all $k\geq K$, that is there exists an instant $t=K$
as of which all solitons are well-separated. 
\end{thm}
By ``well-separated'' we mean that 
the solitons are ordered by their speeds, fastest on the right. Since
each soliton moves with a speed equal to its length, once this ordering
is established no further soliton interactions will take place. 

An elementary proof of Theorem \ref{thm:SolitonContents_Tinfinity} will 
be given in Section 2.3.

\subsection{Elimination ``with riggings''}
\label{subsec:elimination_with_riggings}
Before considering the relationship with rigged configurations, let us recall the 
notion of so-called \textit{``0-solitons''}, introduced in \cite{YYT}. 
As an example, consider the sequence
\begin{equation}
\begin{array}{rcl@{$\,$}l@{$\,$}l@{$\,$}l@{$\,$}l@{$\,$}l@{$\,$}l@{$\,$}l@{$\,$}l}
\bold{u} &=& 
0\,1\,1 & \arc & 0 & \arc & 0\,1\,1 & \arc & 1 & \arc & 0\,0\cdots,\\
&&      & d_4  &   & d_3  &         & d_2  &   & d_1 & 
\end{array}
\end{equation}
where 
$\{d_1,d_2,d_3,d_4\}=\{14,11,6,3\}$.
Applying $\phi_{14}$ to $\bold{u}$, 
we have $\phi_{14}(\bold{u})=0111001001110100\cdots$, and 
$\mathrm{des}\left(\phi_{14}(\bold{u})\right)=\mathrm{des}(\bold{u})=4$. 
On the other hand, applying $\phi_{11}$ to $\bold{u}$, 
we have $\phi_{11}(\bold{u})=0111001001111000\cdots$, and thus
$\mathrm{des}\left(\phi_{11}(\bold{u})\right)=3<\mathrm{des}(\bold{u})=4$. 
Furthermore, the 10 eliminations at $d_3=6$ and $d_4=3$ give 
$\mathrm{des}\left(\phi_{6}(\bold{u})\right)=3$ and 
$\mathrm{des}\left(\phi_{3}(\bold{u})\right)=4$, respectively. 
Let us put vertical lines in 
$\Phi_{10}(\bold{u})=0110011100\cdots$ 
at the positions that correspond to the 10 pairs at 
$d_2=11$ and $d_3=6$, where the corresponding 
10-eliminations lower the descent number:
\begin{equation}
0\,1\,1\,0\,|\,0\,1\,1\,|\,1\,0\,0\,\cdots.
\label{e.g.:0-soliton}
\end{equation}
The vertical lines in \eqref{e.g.:0-soliton} are examples of 
``0-solitons'' in the sense of \cite{YYT}. 
Although the left-most 0-soliton in \eqref{e.g.:0-soliton} 
lies between $(\Phi_{10}\left(\bold{u})\right)_3$ and 
$\left(\Phi_{10}(\bold{u})\right)_4$, we shall say that it is
located at position 3. Adhering to the same convention, the other
0-soliton is then located at position 6. 
For a sequence $\bold{u}\in\mathcal{U}_N$ with 
$\mathrm{Des}(\bold{u})=\left\{d_1>\cdots >d_N\right\}$ and 
$\mathrm{Asc}(\bold{u})=\left\{a_1>\cdots >a_N\right\}$, 
such 0-solitons may appear if 
$d_j=a_j+1$ (for some $j=1,\ldots,N$) or $d_j=a_{j-1}-1$ (for some $j=2,\ldots,N$).
Remark however that these are necessary, but not always sufficient 
conditions for the soliton number to change under $10$-elimination. 
For example, consider the following sequence with
$\{a_1,a_2\}=\{4,1\}$ and $\{d_1,d_2\}=\{5,3\}$: 
\begin{equation}
\begin{array}{ccccccccccc}
\bold{u}(4,1;5,3) &=& 0 & 0 & 1 & 1 & 0 & 1 & 0 & 0 & \cdots.\\
         & &   &a_2&&d_2&a_1&d_1&   &   & 
\end{array}
\end{equation}
Applying $\phi_5$ and $\phi_3$ successively, we have
\begin{equation}
\begin{aligned}
\bold{u}(4,1;5,3) &= 
\,0\;0\;1\arc\:\arc 0\,\cdots, 
\qquad \mathrm{des}\left(\bold{u}(4,1;5,3)\right)=2,\\
\phi_5\left(\bold{u}(4,1;5,3)\right) &= 
\,0\;0\;1\arc |\:0\,\cdots, \qquad
\left(\mathrm{des}\circ\phi_5\right)\left(\bold{u}(4,1;5,3)\right)=1,\\
\left(\phi_3\circ\phi_5\right)\left(\bold{u}(4,1;5,3)\right) &= 
\,0\;0\;1\,|\,0\,\cdots, \qquad
\left(\mathrm{des}\circ\phi_3\circ\phi_5\right)\left(\bold{u}(4,1;5,3)\right)=1,
\end{aligned}
\end{equation}
{}from which it is clear that although
$d_2=a_1-1=3$, the elimination $\phi_3$ does not give rise to a 0-soliton.

To record the positions of 0-solitons 
arising from $\bold{u}\in\mathcal{U}_N$ by the elimination process, we
define a series of increasing integer sequences 
$I_0=\emptyset\subseteq I_1\subseteq\cdots\subseteq I_N$. 
Setting $\bold{u}_0 = \bold{u}$ 
and $\left\{d_1,\ldots,d_N\right\}=\mathrm{Des}(\bold{u})$, 
we define $\bold{u}_1,\ldots,\bold{u}_N$ and 
$\emptyset=I_0\subseteq I_1\subseteq\cdots\subseteq I_N\subseteq \mathrm{Des}(\bold{u})$
by the following recursion relations ($j=1,2,\ldots,N$): 
\begin{align}
\bold{u}_{j}&=\phi_{d_{j}}(\bold{u}_{j-1}), \\
I_{j} & = 
\begin{cases}
I_{j-1}\cup\left\{d_j\right\} & 
\mbox{if \ } \mathrm{des}(\bold{u}_{j-1})>\mathrm{des}(\bold{u}_{j}),\\
I_{j-1} & 
\mbox{if \ } \mathrm{des}(\bold{u}_{j-1})=\mathrm{des}(\bold{u}_{j}),
\end{cases}
\label{recusionRel_Ij}
\end{align}
i.e., we add position data whenever a 10-pair gives rise to a
0-soliton.
We would like to emphasize that the $d_j$ used in the
recurrence \eqref{recusionRel_Ij} are
defined on the original binary sequence $\bold{u}_0=\bold{u}$, 
\textit{not} on $\bold{u}_j$ ($j\geq 1$).

We then prepare a piecewise-linear function 
$f_{\bm{c}}:\mathbb{Z}\to\mathbb{Z}$ associated with 
an integer sequence 
$\bm{c}= \{c_1>\dots>c_N\geq 0\}\in\mathbb{Z}^{N}$. 
For $n<0$, we set $f_{\bm{c}}(n) = n$. 
For $n\geq 0$, $f_{\bm{c}}(n)$ is obtained from $f_{\bm{c}}(n-1)$ as 
\begin{equation}
f_{\bm{c}}(n) = 
\begin{cases}
f_{\bm{c}}(n-1) & 
(n\in\left\{c_1,\ldots,c_N,c_1+1,\ldots,c_N+1\right\}),\\
f_{\bm{c}}(n-1)+1 & (\text{otherwise}).
\end{cases}
\label{def:f_c}
\end{equation}
An example of the action of this map (for $N=4$, $\bm{c}=\{11,8,6,2\}$) is given in Figure
\ref{fig:eg_fc(n)}. 
\begin{figure}[htbp]
\begin{equation*}
\begin{array}{cccccccccccccccccccc}
&&& & & & c_4 & & & & c_3 & & c_2 & & & c_1
\\
n &:& \cdots & -1 &
0&1&2&3&4&5&6&7&8&9&10&11&12&13&14&\cdots\\
f_{\bm{c}}(n) & : & \cdots & -1 &
0&1&1&1&2&3&3&3&3&3&4&4&4&5&6&\cdots
\end{array}
\end{equation*}
\caption{Example of the renumbering 
in \eqref{def:f_c} ($N=4$, $c_1=11$, $c_2=8$, 
$c_3=6$, $c_4=2$)}
\label{fig:eg_fc(n)}
\end{figure}

For $\bold{u}\in\mathcal{U}_N$, 
let $I_N=\left\{i_1>i_2>\cdots\right\}$ 
be the integer sequence defined above. 
Define a map $\rho_{10}:\mathcal{U}_N\to \mathbb{Z}_{\geq 0}^{|I_N|}$ as 
\begin{equation}
\rho_{10}\left(\bold{u}\right)
=f_{\mathrm{Des}(\bold{u})}\left(I_N\right), 
\end{equation}
where we use the notation
$f_{\bm{c}}\left(I_N\right)=
\left\{f_{\bm{c}}(i_1)\geq f_{\bm{c}}(i_2)\geq\cdots\right\}$. 

The map $\Phi_{10}$ might seem irreversible, but in fact, one can
reconstruct the original sequence $\bold{u}$ from 
$\Phi_{10}(\bold{u})$ and $\rho_{10}(\bold{u})$.
To this end, we introduce 
$\psi^{10}_{n}:\mathcal{U}\to\mathcal{U}$ (insertion of 
a $10$ pair between $(\bold{u})_n$ and $(\bold{u})_{n+1}$) as 
\begin{equation}
\left(\psi^{10}_{n}\left(\bold{u}\right)\right)_k := 
\begin{cases}
\left(\bold{u}\right)_k & (k\leq n),\\
1 & (k=n+1),\\
0 & (k=n+2),\\
\left(\bold{u}\right)_{k-2} & (k\geq n+3).
\end{cases}
\end{equation}
Given $\bold{u}\in\mathcal{U}_N$ 
and a non-increasing integer sequence 
$J=\left\{j_1\geq\cdots\geq j_{\ell}\geq 0\right\}$, 
we define another non-increasing sequence 
$I=\left\{i_1\geq\cdots\geq i_{N+\ell}\geq 0\right\}$
by reordering the concatenation of 
$\mathrm{Des}(\bold{u})$ and $J$.
We then define $\Psi_{10}(\bold{u},J)$ as
\begin{equation}
\Psi_{10}\left(\bold{u},J\right)
= \left(\psi^{10}_{i_1}\circ\cdots\circ\psi^{10}_{i_{N+\ell}}\right)
\left(\bold{u}\right).
\end{equation}
It is obvious from the definitions that $\Psi_{10}\left(
\Phi_{10}(\bold{u}), \rho_{10}(\bold{u})
\right) = \bold{u}$ for all $\bold{u}\in\mathcal{U}$. 
This means that $\left(\Phi_{10}(\bold{u}),\rho_{10}(\bold{u})\right)$ 
carries enough information to reconstruct the original data.
We denote by 
\begin{equation}
\label{dfn:Ji(10)}
J_i^{(10)} =\left\{J^{10}_{i,1}\geq J^{10}_{i,2}\geq\cdots\right\}
= \left(\rho_{10}\circ\Phi_{10}^{i-1}\right)(\bold{u})
\quad (i=1,2,\ldots),
\end{equation}
the non-increasing integer sequence 
that labels the positions of 0-solitons in $\Phi_{10}^{i}(\bold{u})$, 
and call it the ``$i$th $10$-rigging''.

The map $\rho_{01}$ can be defined in the same manner. 
For $\bold{u}\in\mathcal{U}_N$,
define a series of increasing integer sequences 
$I'_0=\emptyset\subseteq I'_1\subseteq\cdots\subseteq I'_N$
as follows; 
Setting $\bold{u}'_0 = \bold{u}$ 
and $\left\{a_1,\ldots,a_N\right\}=\mathrm{Asc}(\bold{u})$, 
we define $\bold{u}'_1,\ldots,\bold{u}'_N$ and 
$\emptyset=I'_0\subseteq I'_1\subseteq\cdots\subseteq I'_N\subseteq 
\mathrm{Asc}(\bold{u})$
by the following recursion relations
($j=1,2,\ldots,N$): 
\begin{align}
\bold{u}'_{j}&=\phi_{a_{j}}(\bold{u}'_{j-1}), 
\label{recursion:u'}\\
I'_{j} & = 
\begin{cases}
I'_{j-1}\cup\left\{a_j\right\} & 
\mbox{if \ } \mathrm{asc}(\bold{u}'_{j-1})>\mathrm{asc}(\bold{u}'_{j}),\\
I'_{j-1} & 
\mbox{if \ } \mathrm{asc}(\bold{u}'_{j-1})=\mathrm{asc}(\bold{u}'_{j}).
\end{cases}
\label{recusionRel_I'j}
\end{align}
Note that, as was the case for the $d_j$ in \eqref{recusionRel_Ij}, 
the $a_j$ in \eqref{recusionRel_I'j} 
are defined on the original binary sequence $\bold{u}'_0=\bold{u}$, 
\textit{not} on $\bold{u}'_j$ ($j\geq 1$).
The map $\rho_{01}:\mathcal{U}_N\to\mathbb{Z}_{\geq 0}^{|I'_N|}$ is defined by
\begin{equation}
\rho_{01}\left(\bold{u}\right)
=f_{\mathrm{Asc(\bold{u})}}\left(I'_N\right).
\end{equation}
The 01-insertion map $\Psi_{01}$ can be defined analogously, such that
$\Psi_{01}\left(\Phi_{01}(\bold{u}), \rho_{01}(\bold{u})
\right) = \bold{u}$ for all $\bold{u}\in\mathcal{U}^{(1)}$. 
We define the  ``$i$th $01$-rigging'' as 
\begin{equation}
\label{dfn:Ji(01)}
 J_i^{(01)} = \left(\rho_{01}\circ\Phi_{01}^{i-1}\right)(\bold{u}). 
\end{equation}
As is to be expected, the $i$th $01$-rigging $J_i^{(01)}$ is in fact related to 
$J_i^{(10)}$. To see this, 
we must first study the properties of $\rho_{01}$ 
and $\rho_{10}$. Clearly,
\begin{equation}
\rho_{01}\circ\Lambda = \sigma\circ\rho_{01},\quad
\rho_{10}\circ\Lambda = \sigma\circ\rho_{10},
\label{rho_Lambda=sigma_rho}
\end{equation}
where the map $\sigma$ defines a uniform upshift on integer sequences
$\left\{n_1,n_2,\ldots\right\}$:
\begin{equation}
\sigma:\left\{n_1,n_2,\ldots\right\} \mapsto
\left\{n_1+1,n_2+1,\ldots\right\}. 
\end{equation}
Furthermore, the following lemma plays a crucial role in our approach.
\begin{lemma}
\label{lemma:rho10=sigma_rho01}
$\rho_{10} = \sigma\circ\rho_{01}$ on $\mathcal{U}$.
\end{lemma}
To prove Lemma \ref{lemma:rho10=sigma_rho01}, 
we must prepare one more lemma. 
\begin{lemma}\label{Lemma:fb-fa}
Suppose $\bm{a}=(a_1,\ldots,a_N)$ and
$\bm{d}=(d_1,\ldots,d_N)\in(\mathbb{Z}_{\geq 0})^N$
satisfy the 
interlacing property \eqref{InterlacingProperty}. 
Then
\begin{equation}
f_{\bm{d}}(n)-f_{\bm{a}}(n)=
\begin{cases}
2 & (a_k<n<d_k\, (k=1,\ldots,N)),\\
1 & (n=a_k \text{~or~} d_k\, (k=1,\ldots,N)),\\
0 & (n<a_N,\; d_k<n<a_{k-1}\, (k=2,\ldots,N),\; d_1<n).
\end{cases}
\end{equation} 
\end{lemma}
Lemma \ref{Lemma:fb-fa} is best explained on an example
(e.g. Figure \ref{Example:Lemma:fb-fa}).
\begin{figure}[htbp]
\[
\begin{array}{ccccccccccccccccccc}
n&: &0&1&2&3&4&5&6&7 & 8 &9& \mbox{\small$10$}
&\mbox{\small$11$} &\mbox{\small$12$} &\mbox{\small$13$}
&\mbox{\small$14$}& \mbox{\small$15$}& \cdots\\
 & & & & \mbox{\small$a_4$} & \mbox{\small$d_4$} & 
 & \mbox{\small$a_3$} & & & \mbox{\small$d_3$}  & \mbox{\small$a_2$}
 & \mbox{\small$d_2$} & \mbox{\small$a_1$} & & \mbox{\small$d_1$} & & &
\\
\bold{u}&: & 0 & 0 & 0 & 1 & 0 & 0 & 1 & 1 & 1 
& 0 & 1 & 0 & 1 & 1 & 0 & 0 & \cdots\\
f_{\bm{a}}(n) &:& 0 & 1 & 1 & 1 & 2 & 2 & 2 & 3 & 4 & 4 & 4 & 4 & 4 
& 5 & 6 & 7 & \cdots\\
f_{\bm{d}}(n) &:& 
0 & 1 & 2 & 2 & 2 & 3 & 4 & 5 & 5 & 5 & 5 & 5 & 6 & 6 & 6 & 7 & \cdots
\end{array}
\]
\caption{Example of Lemma \ref{Lemma:fb-fa}}
\label{Example:Lemma:fb-fa}
\end{figure}

\begin{proof}[Proof of Lemma \ref{lemma:rho10=sigma_rho01}]
 Given $\bold{u}\in\mathcal{U}$, we define 
$\bm{a}=\{a_1>a_2>\cdots\}=\mathrm{Asc}(\bold{u})$, 
$\bm{d}=\{d_1>d_2>\cdots\}=\mathrm{Des}(\bold{u})$. 
Successive 10-patterns in $\bold{u}$ can be categorized into $4$ types
($n\in\mathbb{N}$):
\begin{quote}
(I) \ $00(10)^n0$ \qquad
(II) \ $00(10)^n11$ \qquad
(III) \ $1(10)^n0$ \qquad
(IV) \ $1(10)^n11$.
\end{quote}

\noindent\underline{Case I: $00(10)^n0=0(01)^n00$}

We assume the position of the leftmost ``01'' is $a_k$. 
\begin{equation}
\begin{array}{c@{}c@{}c@{}c}
& a_k & & \\
0 & 0 & (10)^n & 0
\end{array}
\begin{array}{c} \\ = \end{array}
\begin{array}{cccccccccccc}
 & a_k & d_k & a_{k-1} & d_{k-1} & & \cdots & & & d_{k-n+1}& & \\
0 & 0 & 1 & 0 & 1 & 0 & \cdots & 1 & 0 & 1 & 0 & 0
\end{array}
\end{equation}
Applying the 10-eliminations
$\phi_{d_k}\circ\cdots\circ\phi_{d_{k-n+1}}$, 
one sees that 
$I_N\supseteq\left\{d_{k},d_{k-1},\ldots,d_{k-n+1}\right\}$.
Similarly we have $I'_N\supseteq
\left\{a_{k},a_{k-1},\ldots,a_{k-n+1}\right\}$. 
It follows from Lemma \ref{Lemma:fb-fa} that 
$f_{\bm{d}}(d_k)-f_{\bm{a}}(d_k)=1$, and 
\begin{equation}
\begin{aligned}
f_{\bm{a}}(d_k) &= f_{\bm{a}}(a_k)=f_{\bm{a}}(a_{k-1})=\cdots =
f_{\bm{a}}(a_{k-n+1}),\\
f_{\bm{d}}(d_k) &= f_{\bm{d}}(d_{k-1})=\cdots= f_{\bm{d}}(d_{k-n+1}).
\end{aligned}
\end{equation}
{}from \eqref{def:f_c}. Thus we have
\begin{equation}
\left\{f_{\bm{d}}(d_k),f_{\bm{d}}(d_{k-1}),\ldots,f_{\bm{d}}(d_{k-n+1})\right\}=\left\{f_{\bm{a}}(a_k),f_{\bm{a}}(a_{k-1}),\ldots,f_{\bm{a}}(a_{k-n+1})\right\}
+\{1,1,\ldots,1\}.
\end{equation}

\noindent\underline{Case II: $00(10)^n11 = 0(01)^{n+1}1$}

As above, we assume the position of the leftmost ``01'' is $a_k$. 
\begin{equation}
\begin{array}{c@{}c@{}c@{}c@{}c@{}c}
 & a_k & & & \\
0 & 0 & (10)^n & 1 & 1
\end{array}
\begin{array}{c} \\ = \end{array}
\begin{array}{ccccccccccccc}
& a_k & d_k & a_{k-1} & d_{k-1} &  & \cdots & & & d_{k-n+1}& 
a_{k-n}  & & \\
0 & 0 & 1 & 0 & 1 & 0 & \cdots & 1 & 0 & 1 & 0 & 1 & 1
\end{array}
\end{equation}
In this case, 
\begin{equation}
I_N\supseteq\left\{
d_k,d_{k-1},\ldots,d_{k-n+1}
\right\},\quad 
I'_N\supseteq\left\{
a_{k-1},a_{k-2},\ldots,a_{k-n}
\right\} \quad \left(I'_N\not\ni a_k\right)
\end{equation}
and 
\begin{equation}
\begin{aligned}
f_{\bm{a}}(d_k)
&=f_{\bm{a}}(a_k)=f_{\bm{a}}(a_{k-1})=\cdots = f_{\bm{a}}(a_{k-n}),\\
f_{\bm{d}}(d_k)&=f_{\bm{d}}(d_{k-1})=\cdots = f_{\bm{d}}(d_{k-n+1}).
\end{aligned}
\end{equation}
{}From Lemma \ref{Lemma:fb-fa}, we have
$f_{\bm{d}}(d_k)-f_{\bm{a}}(d_k)=1$. 
It follows that 
\begin{equation}
\left\{
f_{\bm{d}}(d_{k}),f_{\bm{d}}(d_{k-1}),\ldots,f_{\bm{d}}(d_{k-n+1})\right\}=\left\{f_{\bm{a}}(a_{k-1}),f_{\bm{a}}(a_{k-2}),\ldots,f_{\bm{a}}(a_{k-n})
\right\}+\{1,1,\ldots,1\}.
\end{equation}
The remaining two cases can be proved in similar way.
\end{proof}

Now we can explain the relation between $J^{10}_i$ and $J^{01}_i$.
\begin{thm}
\label{thm:J10=s^i(J01)}
$J^{10}_i=\sigma^i\left(J^{01}_i\right)$ 
($i=1,2,\ldots$).
\end{thm}
\begin{proof}
The desired relation is a direct consequence of 
the definitions \eqref{dfn:Ji(10)}, \eqref{dfn:Ji(01)}, 
the relations \eqref{Phi10=LambdaPhi01=Phi01Lambda}, 
\eqref{rho_Lambda=sigma_rho}, 
and Lemma \ref{lemma:rho10=sigma_rho01}. 
\end{proof}

\subsection{Rigged configurations and linearization}
In order to make this paper self-contained, we briefly review 
the definition of rigged configurations for the $\mathfrak{sl}_2$-case, 
following \cite{KOSTY,Reynolds,SchillingReview}. 
Consider a partition $\mu=\left(\mu_1,\ldots,\mu_L\right)$ and 
its conjugate $\lambda={}^t\mu=(\lambda_1, \ldots, \lambda_\ell)$, where $\ell=\mu_1$ (and $\lambda_1=L$). Define 
$m_j=\lambda_j-\lambda_{j+1}$ ($j=1,2,\ldots,\ell$), with $\lambda_{\ell+1}=0$, i.e., $m_j$ counts the number of parts of size $\mu_{1+\lambda_{j+1}}$ in the partition $\mu$. 
A rigging associated to a partition $\mu$ is a collection of (collections of)
integers $J=\{J_1,J_2,\ldots,J_\ell\}$, with $J_i=\{J_{i,1},\ldots,J_{i,m_i}\}$ 
$(i=1,\ldots,\ell)$, as shown in Figure \ref{fig:rigged_config}. 

\begin{figure}[htbp]
\begin{center}
\begin{tikzpicture}
\draw[thick] (0,0) rectangle (0.5,6);
  \node (lambda1) at (0.25,5.2) {\rotatebox{90}{\scriptsize $\lambda_1$}};
  \coordinate (B1) at (0.25,0);
  \coordinate (T1) at (0.25,6);
    \draw (0.4,1.1) node[anchor=west] {\scriptsize $J_{1,1}$};
    \draw (0.6,0.75) node[anchor=west] {\scriptsize $\vdots$};
    \draw (0.4,0.2) node[anchor=west] {\scriptsize $J_{1,m_1}$};
\draw[thick] (0.5,1.4) rectangle (1,6);
  \node (lambda2) at (0.75,5.2) {\rotatebox{90}{\scriptsize $\lambda_2$}};
  \coordinate (B2) at (0.75,1.4);
  \coordinate (T2) at (0.75,6);
    \draw (0.9,2.7) node[anchor=west] {\scriptsize $J_{2,1}$};
    \draw (1.1,2.3) node[anchor=west] {\scriptsize $\vdots$};
    \draw (0.9,1.6) node[anchor=west] {\scriptsize $J_{2,m_2}$};
\draw[thick] (1,3) rectangle (1.5,6);
  \node (lambda3) at (1.25,5.2) {\rotatebox{90}{\scriptsize $\lambda_3$}};
  \coordinate (B3) at (1.25,3);
  \coordinate (T3) at (1.25,6);
\draw[thick] (1.5,3.6) rectangle (2,6);
  \draw (1.75,5.2) node {\tiny $\cdots$};
\draw[thick] (2,4) rectangle (2.5,6);
  \draw (2.25,5.2) node {\tiny $\cdots$};
\draw[thick] (2.5,4.4) rectangle (3,6);
  \node (lambdal) at (2.75,5.2) {\rotatebox{90}{\scriptsize $\lambda_\ell$}};
  \coordinate (Bl) at (2.75,4.4);
  \coordinate (Tl) at (2.75,6);
    \draw (2.9,5.7) node[anchor=west] {\scriptsize $J_{\ell,1}$};
    \draw (3.1,5.3) node[anchor=west] {\scriptsize $\vdots$};
    \draw (2.9,4.6) node[anchor=west] {\scriptsize $J_{\ell,m_\ell}$};
\draw[->] (lambda1) -- (T1);
\draw[->] (lambda1) -- (B1);
\draw[->] (lambda2) -- (T2);
\draw[->] (lambda2) -- (B2);
\draw[->] (lambda3) -- (T3);
\draw[->] (lambda3) -- (B3);
\draw[->] (lambdal) -- (Tl);
\draw[->] (lambdal) -- (Bl);
\draw (-0.5,5.7) node {\scriptsize $\mu_1\rightarrow$};
\draw (-0.5,5.3) node {\scriptsize $\mu_2\rightarrow$};
\draw (-0.5,4.9) node {\scriptsize $\mu_3\rightarrow$};
\draw (-0.5,3.43) node {\scriptsize $\vdots$};
\draw (-0.5,3) node {\scriptsize $\vdots$};
\draw (-0.5,2.57) node {\scriptsize $\vdots$};
\draw (-0.5,0.2) node {\scriptsize $\mu_L\rightarrow$};
\end{tikzpicture}
\caption{Rigged configuration}
\label{fig:rigged_config}
\end{center}
\end{figure}
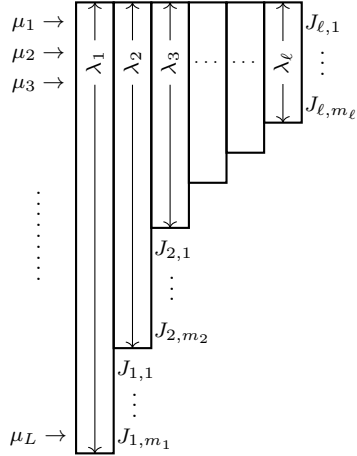

For $\bold{u}\in\mathcal{U}^+$, 
the rigged configuration 
$\left\{\mu,J\right\}$ associated with $\bold{u}$ 
is defined as follows. We assume that $u_M=1$ and 
$u_j=0$ for all $j>M$. For $j=1,2,\ldots$, denote by
$\check{\bold{u}}_j=u_1u_2\cdots u_j$ 
the finite binary sequence picked out from $\bold{u}=u_0u_1u_2\cdots$
by omitting the first entry $u_0=0$. 
Consider a growing sequence of partitions 
$\emptyset=\mu^{(0)}\subseteq
\cdots\subseteq\mu^{(k)}\subseteq\cdots\subseteq\mu^{(M)}$ 
and associated riggings 
$J^{(k)}=\left\{J^{(k)}_1,J^{(k)}_2,\ldots\right\}$, 
$J^{(k)}_i=\left\{J^{(k)}_{i,1},\ldots,
J^{(k)}_{i,m^{(k)}_i}\right\}$, where 
$k=1,\ldots,M$ and $i=1,\ldots,\mu^{(k)}_1$. 
We introduce the so-called $i$th \textit{vacancy number} 
$p^{(k)}_i$ for the $m_i$ rows with the same size
$\mu_{1+\lambda_{i+1}}=\cdots=\mu_{\lambda_i}$  in the Young diagram
associated to $\mu$, as
\begin{equation}
p^{(k)}_i= k-2
\left(\lambda^{(k)}_1+\cdots+\lambda^{(k)}_i\right).
\label{def:vacancy}
\end{equation}
A row, with rigging $J^{(k)}_{i,j}$, in a rigged configration $\{\mu^{(k)},J^{(k)}\}$ is 
called \textit{singular} if at that value of $j$ the associated rigging 
satisfies $J^{(k)}_{i,j}=p^{(k)}_i$. 
Assume that the subsequence $\check{\bold{u}}_{j-1}=u_1\cdots u_{j-1}$
corresponds to $\left\{\mu^{(j-1)},J^{(j-1)}\right\}$. 
If $u_j=0$, then $\left\{\mu^{(j)},J^{(j)}\right\}
=\left\{\mu^{(j-1)},J^{(j-1)}\right\}$. 
If on the other hand $u_j=1$, then we add a box to the longest singular row
in $\left\{\mu^{(j-1)},J^{(j-1)}\right\}$ and make that row singular
again by affixing the appropriate rigging to it. Note that, by
convention, the empty set (or an empty row in the Young diagram) is
always taken to be singular. 
The rigged configuration 
$\left\{\mu,J\right\}$ that corresponds to $\bold{u}$ 
is given by 
\begin{equation}
\left\{\mu,J\right\}=\left\{\mu^{(M)},J^{(M)}\right\}.
\label{def:RiggedConfiguration}
\end{equation}

\begin{example}
\label{RCeg:00001101100110100,Tinfinity}
$\bold{u}=u_0u_1u_2\cdots =
00001101100110100\cdots\in\mathcal{U}^+$: 
\[
\check{\bold{u}}_1=0
,\quad
\check{\bold{u}}_2=00
,\quad
\check{\bold{u}}_3=000
,\quad
\check{\bold{u}}_4=0001
,\quad
\check{\bold{u}}_5=00011
,\quad\cdots,\quad
\check{\bold{u}}_{14}=00011011001101
\]
The procedure described above works as follows:
\begin{align*}
\ytableausetup{boxsize=1em}
&\begin{CD}
\emptyset @>{u_1=0}>> \emptyset
@>{u_2=0}>> \emptyset @>{u_3=0}>> \emptyset
@>{u_4=1}>> 
\begin{ytableau}
\none[2] & {} & \none[2]
\end{ytableau}
@>{u_5=1}>> 
\begin{ytableau}
\none[1] & {} & {} & \none[1]
\end{ytableau}
@>{u_6=0}>> 
\begin{ytableau}
\none[2] & {} & {} & \none[1]
\end{ytableau}
\end{CD}
\\
&\begin{CD}
@>{u_7=1}>> 
\begin{ytableau}
\none[1] & {} & {} & \none[1]\\
\none[3] & {} & \none[3]\\
\end{ytableau}
@>{u_8=1}>> 
\begin{ytableau}
\none[0] & {} & {} & {} & \none[0]\\
\none[4] & {} & \none[3]\\
\end{ytableau}
@>{u_9=0}>> 
\begin{ytableau}
\none[1] & {} & {} & {} & \none[0]\\
\none[5] & {} & \none[3]\\
\end{ytableau}
@>{u_{10}=0}>> 
\begin{ytableau}
\none[2] & {} & {} & {} & \none[0]\\
\none[6] & {} & \none[3]\\
\end{ytableau}
\end{CD}
\\
&\begin{CD}
@>{u_{11}=1}>> 
\begin{ytableau}
\none[1] & {} & {} & {} & \none[0]\\
\none & {} & \none[5]\\
\none[\raisebox{0.5em}{$5$}] & {} & \none[3]
\end{ytableau}
@>{u_{12}=1}>> 
\begin{ytableau}
\none[0] & {} & {} & {} & \none[0]\\
\none[2] & {} & {} & \none[2]\\
\none[6] & {} & \none[3]
\end{ytableau}
@>{u_{13}=0}>> 
\begin{ytableau}
\none[1] & {} & {} & {} & \none[0]\\
\none[3] & {} & {} & \none[2]\\
\none[7] & {} & \none[3]
\end{ytableau}
@>{u_{14}=1}>> 
\begin{ytableau}
\none[0] & {} & {} & {} & \none[0]\\
\none[2] & {} & {} & \none[2]\\
\none & {} & \none[6]\\
\none[\raisebox{0.5em}{$6$}] & {} & \none[3]
\end{ytableau},
\end{CD}
\end{align*}
where the numbers on the left of the Young diagrams are the vacancy numbers.
The rigged configuration in this case is 
$\left\{\mu,J\right\}
=\left\{(3,2,1,1),\,\left\{
J_1=\{6,3\},\,J_2=\{2\},\,J_3=\{0\}\right\}\right\}$.
\end{example}

We recall two important results in \cite{KS}: 
\begin{thm}[Equivalent to Theorem 4.1 of \cite{KS}]
\label{thm:thm4.1_KS}
For $\bold{u}\in\mathcal{U}^+$, 
the rigging $J$ obtained in the above fashion is related to 
the $i$-th 10-rigging $J^{10}_i$ of \eqref{dfn:Ji(10)} -- i.e. to the 
positions of 0-solitons in $\Phi_{10}^{i}(\bold{u})$ -- 
as 
\begin{equation}
J^{10}_i = \sigma^i\left(J_i\right) \quad (i=1,2,\ldots).
\end{equation}
\end{thm}
\begin{thm}[Equivalent to Theorem 4.2 of \cite{KS}]
\label{thm:thm4.2_KS}
For $\bold{u}\in\mathcal{U}^+$, 
the partition $\mu$ obtained above is 
conjugate to $\lambda=\left(\lambda_1,\lambda_2,\ldots\right)$, where $\lambda_i=\left(\mathrm{asc}\circ
\Phi_{10}^{i-1}\right)(\bold{u})$. 
\end{thm}

Hence, from Theorem \ref{thm:J10=s^i(J01)} and Theorem \ref{thm:thm4.1_KS} it is immediately clear that the above riggings are equal to the 01-rigging defined in \eqref{dfn:Ji(01)}:
\begin{equation}
J_i = J_i^{01} \quad (i=1,2,\ldots).
\label{J=J01}
\end{equation}
Hereafter we shall use the notation $J=\{J_1,J_2,\ldots\}
=\left\{\{J_{1,1},\ldots,J_{1,m_1},\ldots\},
\{J_{2,1},\ldots,J_{2,m_2}\},\ldots\right\}$ to indicate 01-riggings.

Starting from $\bold{u}\in\mathcal{U}^+$, we define 
$\ell$ as the minimum integer that satisfies 
$\mathcal{N}\left(\Phi_{01}^\ell(\bold{u})\right)=0$. 
We define
\begin{equation}
KKR(\bold{u}) :=
\left\{\left(\lambda_1,\cdots,\lambda_\ell\right),
\{J_1,\ldots,J_\ell\}\right\},
\label{def:KKRmap}
\end{equation}
for $\bold{u}\in\mathcal{U}^+$,
which can be represented graphically as in Figure
\ref{fig:rigged_config} since the partition $\lambda$ in \eqref{def:KKRmap} is 
conjugate to $\mu$ in \eqref{def:RiggedConfiguration}
(due to Theorem \ref{thm:thm4.2_KS}). As explained above, the number of elements in each sequence $J_i$ satisfies 
\begin{equation}
\left|J_i\right| = \lambda_{i}-\lambda_{i+1} \quad (i=1,\ldots,\ell-1),
\qquad \left|J_\ell\right| = \lambda_\ell,
\end{equation}
and the total number of riggings is therefore
$\sum_{i=1}^\ell\left|J_i\right|=\lambda_1$, which is of course 
the length of the partition $\mu$, 
as shown in Figure \ref{fig:rigged_config}. 

The following is an example of the relation \eqref{J=J01}.
\begin{example}
\label{eg:00001101100110100,Tinfinity}
$\bold{u}=00001101100110100\cdots\in\mathcal{U}^+$ 
(same as Example \ref{RCeg:00001101100110100,Tinfinity}): 
\begin{align*}
& \left.
\begin{array}{ccc@{\ }c@{\ }c@{\ }c@{\ }c@{\ }c@{\ }c@{\ }c@{\ }c@{\ }c@{\ }c@{\ }c@{\ }c@{\ }c}
\bold{u} & : & 0&0&0&\arczeroone&1&\arczeroone&1&0&\arczeroone&
1&\arczeroone&0&0&\cdots\\
f_{\mathrm{Asc}(\bold{u})}(n) & : & 
0&1&2&2\,\,2&3&3\,\,3&4&5&5\,\,5&6&6\,\,6&7&8&\cdots
\end{array}\right\}
\quad \lambda_1=\mathrm{asc}(\bold{u})=4,\quad
J_1=\rho_{01}(\bold{u}) =\left\{6,3\right\}, 
\\
& \left.
\begin{array}{ccc@{\ }c@{\ }c@{\ }c@{\ }c@{\ }c@{\ }c@{\ }c}
\Phi_{01}(\bold{u}) & : & 0&0&\arczeroone&1&\arczeroone&0&0&\cdots\\
f_{\mathrm{Asc}(\Phi_{01}(\bold{u}))}(n) & : & 
0&1&1\,\,1&2&2\,\,2&3&4&\cdots
\end{array}\right\}
\quad\lambda_2=\mathrm{asc}\left(\Phi_{01}(\bold{u})\right)=2, 
\quad J_2=\rho_{01}\left(\Phi_{01}(\bold{u})\right) =\left\{2\right\},
\\
& \left.
\begin{array}{ccc@{\ }c@{\ }c@{\ }c@{\ }c}
\Phi_{01}^2(\bold{u}) & : & 0&\arczeroone&0&0&\cdots\\
f_{\mathrm{Asc}(\Phi_{01}^2(\bold{u}))}(n) & : & 
0&0\,\,0&1&2&\cdots
\end{array}\right\}
\quad\lambda_3=\mathrm{asc}\left(\Phi_{01}^2(\bold{u})\right)=1, 
\quad J_3=\rho_{01}\left(\Phi_{01}^2(\bold{u})\right) =\left\{0\right\},
\end{align*}
\[
\ytableausetup{boxsize=1em}
KKR(\bold{u})=
\left\{\left(4,2,1\right),\left\{
\left\{6,3\right\},\left\{2\right\},\left\{0\right\}
\right\}\right\}
\; : \quad
\begin{ytableau}
{} & {} & {} & \none[0]\\
{} & {} & \none[2]\\
{} & \none[6]\\
{} & \none[3]
\end{ytableau},
\]
and the resulting rigged configuration clearly coincides with 
that of Example \ref{RCeg:00001101100110100,Tinfinity}.

The time-evolution of this particular state $\bold{u}\in\mathcal{U}^+$  is as follows:
\begin{align*}
\bold{u} &= 000011011001101000000000000\cdots,\\
T(\bold{u}) &= 000000100110010111000000000\cdots,\\
T^2(\bold{u}) &= 000000010001101000111000000\cdots,\\
T^3(\bold{u}) &= 000000001000010110000111000\cdots.
\end{align*}
The asymptotic lengths of the solitons, $(3,2,1,1)$ in order of
 decreasing speed, can be directly observed on the state $T^3(\bold{u})
 $ in which the solitons are well-separated, in the sense of Theorem
 \ref{thm:SolitonContents_Tinfinity}. Viewed as a partition, $(3,2,1,1)$
 corresponds exactly to the Young diagram in the representation of
 $KKR(\bold{u})$ given above.  
\end{example}

The following theorem was originally conjectured in \cite{KOTY} 
and proved in \cite{Takagi,KOSTY,Sakamoto1}. 
\begin{thm}[Linearization of the time-evolution $T$
\cite{KOTY,KOSTY,Takagi,Sakamoto1}]
\label{thm:Liniarization}
For $\bold{u}\in\mathcal{U}^+$, define $\overline{J}_n$ as 
\begin{equation}
\overline{J}_n = \left(\rho_{01}\circ\Phi_{01}^{n-1}\circ T\right)\left(\bold{u}\right).
\end{equation}
Then the relation $\overline{J}_n =\sigma^n\left(J_n\right)$ holds for
$n=1,2,\ldots$. 
\end{thm}

\begin{example}
\label{eg:thm:00001101100110100,Tinfinity}
$\bold{u}=00001101100110100\cdots\in\mathcal{U}^+$ 
(same as Examples \ref{RCeg:00001101100110100,Tinfinity} and \ref{eg:00001101100110100,Tinfinity}): 
\[
\begin{CD}
\bold{u}=
0000110110011010000\cdots
@>{T}>> 
T(\bold{u})=
0000001001100101110\cdots\\
@V{KKR}VV @VV{KKR}V\\
\left\{\begin{array}{l}
\lambda_1 = 4,\; \lambda_2 =2,\; \lambda_3 =1,\\
J_1 = \left\{6,3\right\},\;
J_2 = \left\{2\right\},\;
J_3 = \left\{0\right\}
\end{array}\right\}
@>>{\mbox{\scriptsize $
\begin{array}{l}
\overline{\lambda}_n = \lambda_n,\\
\overline{J}_n = \sigma^n(J_n)
\end{array}$}
}>
\left\{\begin{array}{l}
\overline{\lambda}_1 = 4,\; \overline{\lambda}_2 =2,
\; \overline{\lambda}_3 =1,\\
\overline{J}_1 = \left\{7,4\right\},\;
\overline{J}_2 = \left\{4\right\},\;
\overline{J}_3 = \left\{3\right\}
\end{array}\right\}
\\
\begin{ytableau}
{} & {} & {} & \none[0]\\
{} & {} & \none[2]\\
{} & \none[6]\\
{} & \none[3]
\end{ytableau}
@.
\begin{ytableau}
{} & {} & {} & \none[3]\\
{} & {} & \none[4]\\
{} & \none[7]\\
{} & \none[4]
\end{ytableau}
\end{CD}
\]
\end{example}

Theorem \ref{thm:SolitonContents_Tinfinity} and 
Theorem \ref{thm:Liniarization} can be proved in the following elementary way.
\begin{proof}[Proof of Theorem \ref{thm:Liniarization}]
Because of Lemma \ref{lemma:10-arcs_01-arcs} we have
\begin{equation}
\rho_{10} = \rho_{01}\circ T,
\label{lemma:rho10=rho01_T}
\end{equation}
and using Lemma \ref{lemma:Lambda_T_Phi10=Phi10_T}, 
Lemma \ref{lemma:rho10=sigma_rho01} and the relations
\eqref{rho_Lambda=sigma_rho}, \eqref{lemma:rho10=rho01_T}, 
we obtain
\begin{equation}
\rho_{01}\circ\Phi_{01}^{n-1}\circ T
= \rho_{01}\circ\Lambda^{n-1}\circ T
\circ\Phi_{01}^{n-1} = \sigma^n \circ\rho_{01}\circ\Phi_{01}^{n-1}.
\end{equation}
Thus we have $\overline{J}_n(\bold{u})=\sigma^n\left(J_n(\bold{u})\right)$
for all $\bold{u}\in\mathcal{U}^+$.
\end{proof}

\begin{proof}[Proof of Theorem \ref{thm:SolitonContents_Tinfinity}]
{}From Theorem \ref{thm:Liniarization}, we know that 
there exists an integer $k$ such that the riggings at step $k$, 
$\left\{J_{i,j}(k)\right\}$, satisfy 
\begin{equation}
J_{i+1,m_{i+1}}-J_{i,1}>\mu_{\lambda_{i}},
\end{equation}
and hence (for $\ell$ defined as in Figure \ref{fig:rigged_config}) 
\begin{equation}
J_{\ell,1}\geq J_{\ell,2}\geq\cdots \geq J_{\ell,m_\ell}
> J_{\ell-1,1}\geq \cdots\cdots
\geq J_{2,m_2}
>J_{1,1}\geq\cdots \geq J_{1,m_1}.
\end{equation}
It is clear that repeated application of the 
01-insertion map $\Psi_{01}$ results in 
a state that satisfies the condition 
\eqref{d-a=mu,a-d>=mu}.
\end{proof}

\subsection{Equivalence with Takagi's approach}
In this subsection, we discuss the relationship between 
our approach and Takagi's method for constructing 
$\mathfrak{sl}_2$-rigged configurations \cite{Takagi}. 
First we explain Takagi's construction and show that 
it gives the same data as ours.

Consider a sequence $\bold{u}(a_1,\ldots,a_N;d_1,\ldots,d_N)$ 
where $a_1$, $\ldots$, $a_N$ and $d_1$, $\ldots$, $d_N$ satisfy the interlacing condition 
$0<a_N<d_N<\cdots<a_1<d_1$ (Figure \ref{fig:u(a_1,...,a_N;d_1,...,d_N)}). 
For such a sequence, form the $2\times N$ matrix
\begin{equation}
M = 
\begin{pmatrix}
a_N & a_{N-1} & \cdots & a_2 & a_1\\
d_N & d_{N-1} & \cdots & d_2 & d_1
\end{pmatrix}.
\label{TakagiMatrix}
\end{equation}
The matrix $M'$ is then defined by
\begin{align}
M' &= M-\omega
\begin{pmatrix}
1 & 3 & \cdots & 2N-3 & 2N-1\\
2 & 4 & \cdots & 2N-4 & 2N
\end{pmatrix}
\nonumber\\
&= \begin{pmatrix}
a'_N & a'_{N-1} & \cdots & a'_2 & a'_1\\
d'_N & d'_{N-1} & \cdots & d'_2 & d'_1
\end{pmatrix}
\end{align}
where $a'_j:=a_j-\omega(2N+1-2j)$, $d'_j:=d_j-\omega(2N+2-2j)$ and where $\omega$ is the smallest
positive integer such that 
there exists either a column or an antidiagonal in $M'$ 
that only contains equal entries.
Manifestly, such an equality occurs in the $(N+1-j)$th column if
$a'_j=d'_j$, that is if $d_j-a_j=\omega$, or in the $(N+1-j)$th antidiagonal
if $a'_{j-1}=d'_{j}$, that is if $a_{j-1}-d_j=\omega$. 
Note that $d_j-a_j$ is the length of the $j$th
block of 1s, counted from the right, and $a_{j}-d_{j+1}$ is the length of the $j$th
(finite) block of 0s, as shown in Figure \ref{fig:u(a_1,...,a_N;d_1,...,d_N)}. 
Thus $\omega$ is nothing but the minimum length of all the blocks 
of 1s and 0s:
\begin{equation}
\omega = \min\left\{
d_1-a_1, d_2-a_2, \ldots, d_N-a_N, 
a_1-d_2, a_2-d_3, \ldots, a_{N-1}-d_N
\right\}.
\end{equation}
Then, working from left to right, each pair of coinciding elements in $M'$, 
say both with value $r$, is deleted. 
This value $r$ is the rigging and is given by $r=d'_j$.

After deletion, one obtains a new matrix $M$ and the process is repeated
until all entries in $M$ have been deleted. 
The recorded data  are assembled into the rigged Young diagram
shown in Figure \ref{fig:Takagi} in which at stage $i$ 
in the above process 
$\omega=\omega_i$, and $m_i$ pairs were found with corresponding riggings 
$J_{i,1},J_{i,2},\ldots,J_{i,m_i}$. As shown, at the $i$th stage, exactly $m_i$ 
parts of length $\omega_1+\cdots+\omega_i$ are added to the Young
diagram, from bottom to top.

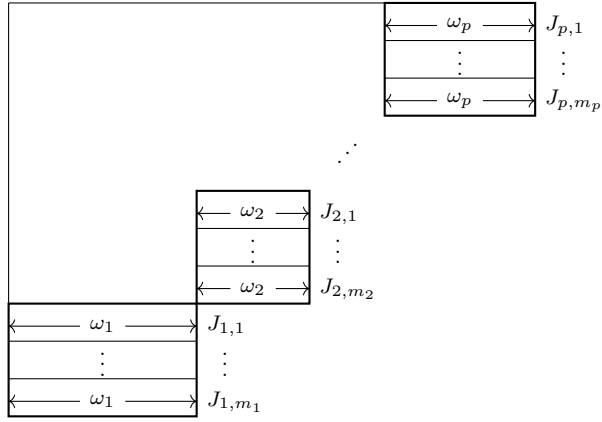
\begin{figure}[htbp]
\begin{center}
\begin{tikzpicture}
\draw[thick] (0,0) rectangle (2.5,1.5);
  \draw (0,0.5) -- (2.5,0.5);
  \draw (0,1) -- (2.5,1);
    \node (k11) at (1.25,1.2) {\scriptsize $\omega_1$};
      \coordinate (L11) at (0,1.2);
      \coordinate (R11) at (2.5,1.2);
      \draw[->] (k11) -- (L11);
      \draw[->] (k11) -- (R11);
    \draw (1.25,0.8) node {\scriptsize $\vdots$};
    \node (k12) at (1.25,0.2) {\scriptsize $\omega_1$};
      \coordinate (L12) at (0,0.2);
      \coordinate (R12) at (2.5,0.2);
      \draw[->] (k12) -- (L12);
      \draw[->] (k12) -- (R12);
    \draw (2.5,1.2) node[anchor=west] {\scriptsize $J_{1,1}$};
    \draw (2.7,0.8) node[anchor=west] {\scriptsize $\vdots$};
    \draw (2.5,0.2) node[anchor=west] {\scriptsize $J_{1,m_1}$};
\draw[thick] (2.5,1.5) rectangle (4,3);
  \draw (2.5,2) -- (4,2);
  \draw (2.5,2.5) -- (4,2.5);
    \node (k21) at (3.25,2.7) {\scriptsize $\omega_2$};
      \coordinate (L21) at (2.5,2.7);
      \coordinate (R21) at (4,2.7);
      \draw[->] (k21) -- (L21);
      \draw[->] (k21) -- (R21);
    \draw (3.25,2.3) node {\scriptsize $\vdots$};
    \node (k22) at (3.25,1.7) {\scriptsize $\omega_2$};
      \coordinate (L22) at (2.5,1.7);
      \coordinate (R22) at (4,1.7);
      \draw[->] (k22) -- (L22);
      \draw[->] (k22) -- (R22);
    \draw (4,2.7) node[anchor=west] {\scriptsize $J_{2,1}$};
    \draw (4.2,2.3) node[anchor=west] {\scriptsize $\vdots$};
    \draw (4,1.7) node[anchor=west] {\scriptsize $J_{2,m_2}$};
    \draw (4.4,3.5) node {\rotatebox{90}{\scriptsize $\ddots$}};
\draw[thick] (5,4) rectangle (7,5.5);
  \draw (5,4.5) -- (7,4.5);
  \draw (5,5) -- (7,5);
    \node (kp1) at (6,5.2) {\scriptsize $\omega_p$};
      \coordinate (Lp1) at (5,5.2);
      \coordinate (Rp1) at (7,5.2);
      \draw[->] (kp1) -- (Lp1);
      \draw[->] (kp1) -- (Rp1);
    \draw (6,4.8) node {\scriptsize $\vdots$};
    \node (kp2) at (6,4.2) {\scriptsize $\omega_p$};
      \coordinate (Lp2) at (5,4.2);
      \coordinate (Rp2) at (7,4.2);
      \draw[->] (kp2) -- (Lp2);
      \draw[->] (kp2) -- (Rp2);
    \draw (7,5.2) node[anchor=west] {\scriptsize $J_{p,1}$};
    \draw (7.2,4.8) node[anchor=west] {\scriptsize $\vdots$};
    \draw (7,4.2) node[anchor=west] {\scriptsize $J_{p,m_p}$};
\draw (0,1.5) -- (0,5.5);
\draw (0,5.5) -- (5,5.5);
\end{tikzpicture}
\caption{Takagi's approach}
\label{fig:Takagi}
\end{center}
\end{figure}

\begin{example}
\label{eg:Takagi'sMmatrix:00001101100110100}
$\bold{u}=00001101100110100\cdots\in\mathcal{U}^+$ 
(same as Example \ref{eg:00001101100110100,Tinfinity}): 
In this case, 
the matrix giving the position of the last entries in each block is
\[
M_0=\begin{pmatrix}
3 & 6 & 10 & 13\\
5 & 8 & 12 & 14
\end{pmatrix}.
\]
Then we have (taking $\omega=1$)
\[
M_0'=M_0- 1\cdot 
\begin{pmatrix}
1 & 3 & 5 & 7\\
2 & 4 & 6 & 8
\end{pmatrix}
=\begin{pmatrix}
2 & \cancel{3} & 5 & \cancel{6}\\
\cancel{3} & 4 & \cancel{6} & 6
\end{pmatrix}
\to \begin{pmatrix}
2 & 5\\
4 & 6
\end{pmatrix} =: M_1, 
\]
and $m_1 = 2$, $\left\{J_{1,1},J_{1,2}\right\}=\left\{6,3\right\}$.
Now (again taking $\omega=1$)
\[
M'_1 = M_1-1\cdot 
\begin{pmatrix}
1 & 3\\ 2 & 4  
\end{pmatrix}
=\begin{pmatrix}
1 & \cancel{2}\\
\cancel{2} & 2
\end{pmatrix}
\to \begin{pmatrix}
1\\
2
\end{pmatrix} =: M_2,
\]
and $m_2=1$, $J_{2,1}=2$. 
Finally (also for $\omega=1$)
\[
M_2' = M_2 -1\cdot 
\begin{pmatrix}
1\\ 2
\end{pmatrix}
=\begin{pmatrix}
\cancel{0}\\ \cancel{0}
\end{pmatrix}
\to \emptyset
\]
and $m_3=1$, $J_{3,1}=0$.
The resulting rigged configuration coincides with 
that of Example \ref{eg:00001101100110100,Tinfinity}.
\end{example}

\begin{prop}
For any sequence $\bold{u}\in\mathcal{U}^+$, 
the rigged configuration resulting  from Takagi's algorithm, 
$\left\{Y(\bold{u}),J(\bold{u})\right\}$, 
where $Y(\bold{u})$ is a Young diagram 
and $J(\bold{u})$ is a set of riggings, 
coincides with the rigged Young diagram obtained 
via our approach, as introduced in subsection (2.2).
\end{prop}
\begin{proof}
Consider a sequence $\bold{u}\in\mathcal{U}_N^+$ with 
$\mathrm{Des}(\bold{u})=\left\{d_1>\cdots >d_N\right\}$ and  
$\mathrm{Asc}(\bold{u})=\left\{a_1>\cdots >a_N\right\}$.
The Takagi matrix $M$ associated with $\bold{u}$ is 
\eqref{TakagiMatrix}.  Define a map $\mathcal{T}$ as 
\begin{equation}
\mathcal{T}(M)=M-\begin{pmatrix}
1 & 3 & \cdots & 2N-1\\
2 & 4 & \cdots & 2N
\end{pmatrix}.
\end{equation}
Define $\omega$ as above and 
\begin{equation}
M^{(j)}=\mathcal{T}^j(M)
=\begin{pmatrix}
d^{(j)}_N & \cdots & d^{(j)}_1\\
a^{(j)}_N & \cdots & a^{(j)}_1
\end{pmatrix}
\quad
(j=1,2,\ldots,\omega).
\end{equation}
For $j=1,2,\ldots,\omega-1$, 
denote by $\bold{u}^{(j)}$
the sequence corresponding to $M^{(j)}$. 
Note that $\mathrm{des}\left(\bold{u}^{(j)}\right)
=\mathrm{asc}\left(\bold{u}^{(j)}\right)=N$, 
$\bold{u}^{(j)} = \Phi_{01}^j\left(\bold{u}\right)$ 
for $j=0,1,\ldots,\omega-1$, 
and the following diagram commutes:
\begin{equation}
\begin{CD}
M @>{\mathcal{T}}>> M^{(1)} @>{\mathcal{T}}>> M^{(2)} 
@>{\mathcal{T}}>> \cdots @>{\mathcal{T}}>> M^{(\omega-1)}\\
@VVV   @VVV         @VVV         @.          @VVV\\
\bold{u} @>>{\Phi_{01}}> \bold{u}^{(1)} @>>{\Phi_{01}}> \bold{u}^{(2)}
@>>{\Phi_{01}}> \cdots @>>{\Phi_{01}}> \bold{u}^{(\omega-1)}
\lefteqn{\; .}
\end{CD}
\end{equation}
Thus we have
\begin{align}
\lambda_j &= \mathrm{asc}\left(\bold{u}^{(j-1)}\right)=N
\quad (j=1,\ldots,\omega), \\
J_j &= \rho_{01}\left(\bold{u}^{(j-1)}\right) = \emptyset
\quad (j=1,\ldots,\omega-1),
\end{align}
and the first non-empty rigging appears for $\bold{u}^{(\omega-1)}$, 
which occurs at 
$d^{(\omega-1)}_j-a^{(\omega-1)}_j=1$ that is $d_j-a_j=\omega$, 
or at $a^{(\omega-1)}_{j-1}-d^{(\omega-1)}_j=1$ that is 
$a_{j-1}-d_j=\omega$. 

Up to now, we have shown that 
the first $\omega$ columns of Figure \ref{fig:rigged_config}, with
 their corresponding riggings, coincide exactly with those of Figure 
 \ref{fig:Takagi}. 
As the matrix $M^{(\omega)}$ contains
equal values in either the same column or the same antidiagonal, 
it does not correspond to any binary sequence. 
Denote by $\widetilde{M}$ the matrix obtained by
deleting all pairs of coinciding elements from $M^{(\omega)}$
and by $\widetilde{\bold{u}}$ the corresponding binary sequence. 
Clearly,  
$\widetilde{\bold{u}}=\Phi_{01}\left(\bold{u}^{(\omega-1)}\right)$ 
and the following diagram commutes:
\begin{equation}
\begin{tikzcd}[column sep=4em]
M^{(\omega-1)} \arrow[r,"\mathcal{T}"]\arrow[d]
& M^{(\omega)} 
\arrow[r,"\oalign{{\scriptsize pairwise}\crcr{\scriptsize deletion}}"] 
& \widetilde{M}\arrow[d]\\
\bold{u}^{(\omega-1)} \arrow[rr,"\Phi_{01}"'] & & \widetilde{\bold{u}}
\lefteqn{\;.}
\end{tikzcd}
\end{equation}
This process can be repeated until all entries in $M$ have been deleted. 
\end{proof}

Combining the results in this section, 
we conclude that the three approaches, the $\mathfrak{sl}_2$-rigged configuration, 
01-elimination with rigging and Takagi's $M$-matrix, are all equivalent. 
We remark that no explicit proof for this equivalence was presented
in \cite{Takagi}. 

\section{Carrier with finite capacity}
\label{sec:finite}
The carrier description of the box-ball system was introduced in
\cite{TM}. In the case where the site capacity is $L$ and the carrier
capacity is $M$, the time-evolution rule is given by
\begin{equation}
\begin{aligned}
u_j^{t+1} &= u_j^t+\min\left\{v_j^t,L-u_j^t\right\}
-\min\left\{u_j^t,M-v_j^t\right\},\\
v_{j+1}^t &= 
v_j^t-\min\left\{v_j^t,L-u_j^t\right\}
+\min\left\{u_j^t,M-v_j^t\right\},
\end{aligned}
\label{BBSC_rule}
\end{equation}
where $u_j^t$ denotes the number of balls in the $j$th box at time $t$, 
and $v_j^t$ the number of balls in the carrier just before the $j$th box 
at time $t$. The rule \eqref{BBSC_rule} can be represented schematically as in 
Figure \ref{fig:BBSC}.
The graphical representation in Figure \ref{fig:Graph_Repr_TM} 
means the following \cite{TM}; 
Assume the carrier carries $v_j^t$ balls before it passes the $j$th box
which contains $u_j^t$ balls. At this stage 
the carrier has $M-v_j^t$ vacant spaces and the $j$th box has 
$L-u_j^t$. When the carrier passes the box, the carrier puts
as many balls as possible into the box, and simultaneously, 
obtains as many balls from the box as possible, i.e., it offloads $\min\left\{v_j^t,L-u_j^t\right\}$ balls into the box 
and receives $\min\left\{u_j^t,M-v_j^t\right\}$ balls from the
box. Obviously, the number of balls is conserved and we have
$v_{j+1}^t -v_j^t = -(u_j^{t+1}-u_j^t)$. One time step in the evolution
corresponds to the carrier moving through all boxes, from left to right,
which gives rise to the equations \eqref{BBSC_rule}.

\begin{figure}[htbp]
\begin{center}
\begin{tikzpicture}
\path (0,1) node[anchor=east] (W) {$v_j^t$};
\path (2,1) node (E)[anchor=west] {$v_{j+1}^t$};
\path (1,2) node (N)[anchor=south] {$u_j^t$};
\path (1,0) node (S)[anchor=north] {$u_j^{t+1}$};
\draw[->,thick](W)--(E);
\draw[->,thick](N)--(S);
\end{tikzpicture}
\caption{Graphical representation of \eqref{BBSC_rule}}
\label{fig:BBSC}
\end{center}
\end{figure}
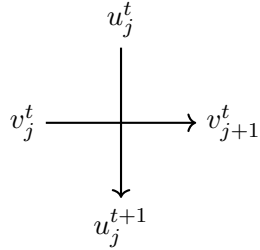
The rule \eqref{BBSC_rule} 
can be obtained from the ultra-discrete limit
of a modified discrete KdV equation \cite{TM,KNW2009}. 
It is equivalent to 
a combinatorial $R$-matrix of $A_1^{(1)}$-type \cite{NY}
as shown in \cite{HHIKTT}.

Hereafter we fix the value of the box capacity at $L=1$ (i.e., $u_j^t=0$ or $1$) and 
consider the transformation $T_M:\mathcal{U}\to\mathcal{U}$; 
$\left\{u_0,u_1,u_2,\ldots\right\}\mapsto
\left\{\overline{u}_0,\overline{u}_1,\overline{u}_2,\ldots\right\}$, 
obtained from the recursion relations
\begin{equation}
\begin{aligned}
\overline{u}_j &= u_j+\min\left\{v_j,1-u_j\right\}
-\min\left\{u_j,M-v_j\right\},\\
v_{j+1} & 
= v_j-\min\left\{v_j,1-u_j\right\}
+\min\left\{u_j,M-v_j\right\}, 
\end{aligned}
\label{def:T_M}
\end{equation}
for the boundary condition $v_0=0$. This process is represented
graphically in Figure \ref{fig:Graph_Repr_TM} and an example for
$T_{M=2}$ is shown in Figure \ref{fig:e.g._M=2}.

Like the Takahashi-Satsuma BBS, 
it is possible to describe rule \eqref{def:T_M} in terms of 
10-arc lines.
\begin{prop}
\label{prop:arc-description_TM}
The time-evolution rule $T_M$ can be described as follows: 
\begin{itemize}
\item[i)] For $\bold{u}\in\mathcal{U}$, define 
$\bold{v}=(v_0,v_1,v_2,\ldots)$ by the initial condition $v_0=0$ and the
recursion relation \eqref{def:T_M}.
\item[ii)] For an integer $n$, if $u_n=1$ and $v_n=M$, then 
underline the ``$1$'' at the $n$th site. 
\item[iii)] Apply the 10-arc line procedure 
of Section \ref{subsec:TimeEvolution_TakahashiSatsuma}
while disregarding all the $\underline{1}$'s.
\item[iv)] Remove all the underlines.
\end{itemize}
\end{prop}

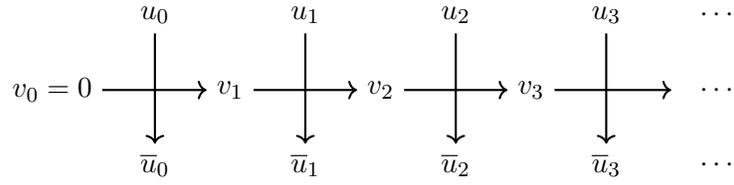
\begin{figure}[htbp]
\begin{center}
\begin{tikzpicture}
\path (0.3,1) node [anchor=east] (V0) {$v_0=0$};
\path (2,1) node (V1) {$v_1$};
\path (4,1) node (V2) {$v_2$};
\path (6,1) node (V3) {$v_3$};
\path (8,1) node (V4) {};
\path (8.5,1) node {$\cdots$};
\draw[->,thick](V0)--(V1);
\draw[->,thick](V1)--(V2);
\draw[->,thick](V2)--(V3);
\draw[->,thick](V3)--(V4);
\path (1,2) node (U00) {$u_0$};
\path (1,0) node (U01) {$\overline{u}_0$};
\draw[->,thick](U00)--(U01);
\path (3,2) node (U10) {$u_1$};
\path (3,0) node (U11) {$\overline{u}_1$};
\draw[->,thick](U10)--(U11);
\path (5,2) node (U20) {$u_2$};
\path (5,0) node (U21) {$\overline{u}_2$};
\draw[->,thick](U20)--(U21);
\path (7,2) node (U30) {$u_3$};
\path (7,0) node (U31) {$\overline{u}_3$};
\draw[->,thick](U30)--(U31);
\path (8.5,0) node {$\cdots$};
\path (8.5,2) node {$\cdots$};
\end{tikzpicture}
\caption{Graphical representation of $T_M$}
\label{fig:Graph_Repr_TM}
\end{center}
\end{figure}

\begin{figure}[htbp]
\begin{center}
\begin{tikzpicture}
\path (0,0.5) node (v0) {$0$};
\path (1,0.5) node (v1) {$0$};
\path (0.5,1) node (u0) {$0$};
\path (0.5,0) node (U0) {$0$};
\draw[->,thick](u0)--(U0);
\draw[->,thick](v0)--(v1);
\path (2,0.5) node (v2) {$1$};
\path (1.5,1) node (u1) {$1$};
\path (1.5,0) node (U1) {$0$};
\draw[->,thick](u1)--(U1);
\draw[->,thick](v1)--(v2);
\path (3,0.5) node (v3) {$2$};
\path (2.5,1) node (u2) {$1$};
\path (2.5,0) node (U2) {$0$};
\draw[->,thick](u2)--(U2);
\draw[->,thick](v2)--(v3);
\path (4,0.5) node (v4) {$2$};
\path (3.5,1) node (u3) {$1$};
\path (3.5,0) node (U3) {$1$};
\draw[->,thick](u3)--(U3);
\draw[->,thick](v3)--(v4);
\path (5,0.5) node (v5) {$1$};
\path (4.5,1) node (u4) {$0$};
\path (4.5,0) node (U4) {$1$};
\draw[->,thick](u4)--(U4);
\draw[->,thick](v4)--(v5);
\path (6,0.5) node (v6) {$2$};
\path (5.5,1) node (u5) {$1$};
\path (5.5,0) node (U5) {$0$};
\draw[->,thick](u5)--(U5);
\draw[->,thick](v5)--(v6);
\path (7,0.5) node (v7) {$2$};
\path (6.5,1) node (u6) {$1$};
\path (6.5,0) node (U6) {$1$};
\draw[->,thick](u6)--(U6);
\draw[->,thick](v6)--(v7);
\path (8,0.5) node (v8) {$2$};
\path (7.5,1) node (u7) {$1$};
\path (7.5,0) node (U7) {$1$};
\draw[->,thick](u7)--(U7);
\draw[->,thick](v7)--(v8);
\path (9,0.5) node (v9) {$1$};
\path (8.5,1) node (u8) {$0$};
\path (8.5,0) node (U8) {$1$};
\draw[->,thick](u8)--(U8);
\draw[->,thick](v8)--(v9);
\path (10,0.5) node (v10) {$0$};
\path (9.5,1) node (u9) {$0$};
\path (9.5,0) node (U9) {$1$};
\draw[->,thick](u9)--(U9);
\draw[->,thick](v9)--(v10);
\path (11,0.5) node (v11) {$0$};
\path (10.5,1) node (u10) {$0$};
\path (10.5,0) node (U10) {$0$};
\draw[->,thick](u10)--(U10);
\draw[->,thick](v10)--(v11);
\end{tikzpicture}
\caption{Example with $M=2$}
\label{fig:e.g._M=2}
\end{center}
\end{figure}
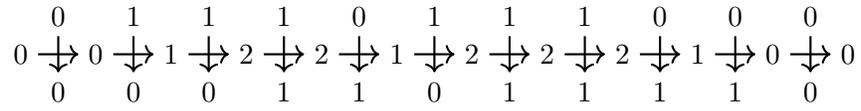

Let us consider the sequence $\bold{u} = 01110111000\cdots$
under the condition $M=2$ as an example. 
Using the second equation in \eqref{def:T_M}, 
one can obtain the sequence
$\bold{v}=\left\{v_0=0,v_1,v_2,\ldots\right\}$ as 
$\bold{v}=001221222100\cdots$, 
and hence $\bold{u}$ is underlined as 
$\bold{u} = 011\underline{1}01\underline{1}
\underline{1}000\cdots$. Then $T_2(\bold{u})$ is 
obtained as in Figure \ref{fig:time-evolution_M=2}, 
which coincides with the result in Figure \ref{fig:e.g._M=2}.

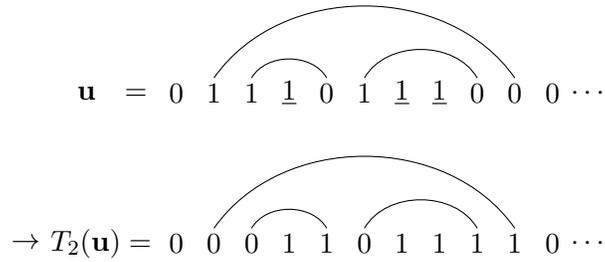
\begin{figure}[htbp]
\begin{center}
\begin{tikzpicture}
\path (1.15,2) node {$\bold{u}$ \ \ $=$};
\path (2,2) node {$0$};
\path (2.5,2) node {$1$};
\path (3,2) node {$1$};
\path (3.5,2) node {$\underline{1}$};
\path (4,2) node {$0$};
\path (4.5,2) node {$1$};
\path (5,2) node {$\underline{1}$};
\path (5.5,2) node {$\underline{1}$};
\path (6,2) node {$0$};
\path (6.5,2) node {$0$};
\path (7,2) node {$0$};
\path (7.5,2) node {$\cdots$};
\draw (3,2.2) to [out=60,in=120] (4,2.2);
\draw (4.5,2.2) to [out=60,in=120] (6,2.2);
\draw (2.5,2.2) to [out=55,in=125] (6.5,2.2);
\path (0,0) node {$\to$};
\path (1,0) node {$T_2(\bold{u})=$};
\path (2,0) node {$0$};
\path (2.5,0) node {$0$};
\path (3,0) node {$0$};
\path (3.5,0) node {$1$};
\path (4,0) node {$1$};
\path (4.5,0) node {$0$};
\path (5,0) node {$1$};
\path (5.5,0) node {$1$};
\path (6,0) node {$1$};
\path (6.5,0) node {$1$};
\path (7,0) node {$0$};
\path (7.5,0) node {$\cdots$};
\draw (3,0.2) to [out=60,in=120] (4,0.2);
\draw (4.5,0.2) to [out=60,in=120] (6,0.2);
\draw (2.5,0.2) to [out=55,in=125] (6.5,0.2);
\end{tikzpicture}
\caption{Example of time-evolution $T_{M=2}$}
\label{fig:time-evolution_M=2}
\end{center}
\end{figure}

Then it follows from Proposition \ref{prop:arc-description_TM} that 
\begin{equation}
\label{TmPhi10=Phi01Tm+1}
T_{M-1}\circ \Phi_{10} = \Phi_{01}\circ T_{M}\quad (M=1,2,\ldots), 
\end{equation}
which is the finite capacity version of \eqref{TPhi=PhiT}.
The relation \eqref{TmPhi10=Phi01Tm+1} is best explained by an example. 
Starting from $\bold{u}=0111100111000$, one has 
$(T_2\circ\Phi_{10})(\bold{u})=(\Phi_{10}\circ T_3)
=000110111$ as shown in Figure \ref{fig:TmPhi10=Phi01Tm+1}.
Note that the operation $\Phi_{10}$ still eliminates
1st arcs in $\bold{u}$, though the 1st arcs do \textit{not} always 
connect adjacent pairs of 1 and 0 (See Figure
\ref{fig:TmPhi10=Phi01Tm+1}). 
Note also that $T_1$ is merely a forward shift $\Lambda$, since 
in the case $M=1$ we have
$\overline{u}_j=v_j=u_{j-1}$. 

\begin{figure}[htbp]
\begin{center}
\begin{tikzpicture}
\path (0,0) node[anchor=north] {$0$};
\path (0.3,0) node[anchor=north] {$0$};
\path (0.6,0) node[anchor=north] {$0$};
\path (0.9,0) node[anchor=north] {$0$};
\path (1.2,0) node[anchor=north] {$1$};
\path (1.5,0) node[anchor=north] {$1$};
\path (1.8,0) node[anchor=north] {$1$};
\path (2.1,0) node[anchor=north] {$0$};
\path (2.4,0) node[anchor=north] {$0$};
\path (2.7,0) node[anchor=north] {$1$};
\path (3.0,0) node[anchor=north] {$1$};
\path (3.3,0) node[anchor=north] {$1$};
\path (3.6,0) node[anchor=north] {$1$};
\draw (0.9,0) to [out=60,in=120] (1.5,0);
\draw (2.4,0) to [out=60,in=120] (3.0,0);
\draw (0.6,0) to [out=60,in=120] (1.8,0);
\draw (2.1,0) to [out=60,in=120] (3.3,0);
\draw (0.3,0) to [out=55,in=125] (3.6,0);
\draw[->] (4,-0.25)--(5,-0.25);
\path (4.5,0) node {\small $\Phi_{01}$};
\path (5.4,0) node[anchor=north] {$0$};
\path (5.7,0) node[anchor=north] {$0$};
\path (6.0,0) node[anchor=north] {$0$};
\path (6.3,0) node[anchor=north] {$1$};
\path (6.6,0) node[anchor=north] {$1$};
\path (6.9,0) node[anchor=north] {$0$};
\path (7.2,0) node[anchor=north] {$1$};
\path (7.5,0) node[anchor=north] {$1$};
\path (7.8,0) node[anchor=north] {$1$};
\draw (6.0,0) to [out=60,in=120] (6.6,0);
\draw (6.9,0) to [out=60,in=120] (7.5,0);
\draw (5.7,0) to [out=60,in=120] (7.8,0);
\draw[->] (1.8,2)--(1.8,1);
\path (1.5,1.5) node {\small $T_3$};
\draw[->] (6.6,2)--(6.6,1);
\path (6.9,1.5) node {\small $T_2$};
\path (0,2.7) node[anchor=north] {$0$};
\path (0.3,2.7) node[anchor=north] {$1$};
\path (0.6,2.7) node[anchor=north] {$1$};
\path (0.9,2.7) node[anchor=north] {$1$};
\path (1.2,2.7) node[anchor=north] {$\underline{1}$};
\path (1.5,2.7) node[anchor=north] {$0$};
\path (1.8,2.7) node[anchor=north] {$0$};
\path (2.1,2.7) node[anchor=north] {$1$};
\path (2.4,2.7) node[anchor=north] {$1$};
\path (2.7,2.7) node[anchor=north] {$\underline{1}$};
\path (3.0,2.7) node[anchor=north] {$0$};
\path (3.3,2.7) node[anchor=north] {$0$};
\path (3.6,2.7) node[anchor=north] {$0$};
\draw (0.9,2.7) to [out=60,in=120] (1.5,2.7);
\draw (2.4,2.7) to [out=60,in=120] (3.0,2.7);
\draw (0.6,2.7) to [out=60,in=120] (1.8,2.7);
\draw (2.1,2.7) to [out=60,in=120] (3.3,2.7);
\draw (0.3,2.7) to [out=55,in=125] (3.6,2.7);
\draw[->] (4,2.45)--(5,2.45);
\path (4.5,2.7) node {\small $\Phi_{10}$};
\path (5.4,2.7) node[anchor=north] {$0$};
\path (5.7,2.7) node[anchor=north] {$1$};
\path (6.0,2.7) node[anchor=north] {$1$};
\path (6.3,2.7) node[anchor=north] {$\underline{1}$};
\path (6.6,2.7) node[anchor=north] {$0$};
\path (6.9,2.7) node[anchor=north] {$1$};
\path (7.2,2.7) node[anchor=north] {$\underline{1}$};
\path (7.5,2.7) node[anchor=north] {$0$};
\path (7.8,2.7) node[anchor=north] {$0$};
\draw (6.0,2.7) to [out=60,in=120] (6.6,2.7);
\draw (6.9,2.7) to [out=60,in=120] (7.5,2.7);
\draw (5.7,2.7) to [out=60,in=120] (7.8,2.7);
\end{tikzpicture}
\caption{Example of the relation \eqref{TmPhi10=Phi01Tm+1}}
\label{fig:TmPhi10=Phi01Tm+1}
\end{center}
\end{figure}
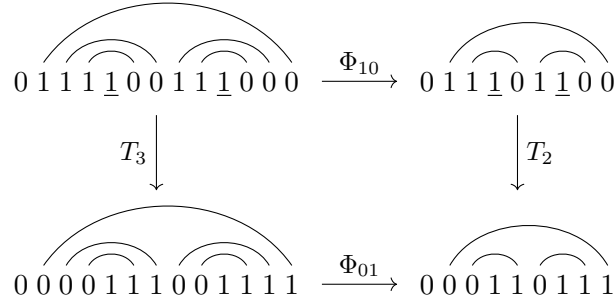

We can now consider invariants with respect to the $T_M$ evolution. 
\begin{thm}
\label{thm:1st_CQ}
For $\bold{u}\in\mathcal{U}^+$, define $\ell$ as the 
minimal integer that satisfies
$\mathcal{N}\left(\Phi_{01}^{\ell}\right)(\bold{u})=0$. 
For $j=1,2,\ldots,\min\{\ell,M\}$, 
define $\lambda_j(\bold{u})$ as $\lambda_j(\bold{u})
:=\left(\mathrm{asc}\circ \Phi_{01}^{j-1}\right)(\bold{u})$. 
Then we have 
$(\lambda_j\circ T_M)(\bold{u})=\lambda_j(\bold{u})$.
\end{thm}
Theorem \ref{thm:1st_CQ} is a finite-carrier version of 
Theorem \ref{thm:invariants0}, and 
can be proved in the same manner by using 
\eqref{TmPhi10=Phi01Tm+1} instead of \eqref{TPhi=PhiT}.
\begin{proof}
The desired relation follows from
\begin{equation}
\Lambda^k\circ T_{M-k}\circ\Phi_{01}^k
=\Phi_{01}^k\circ T_M
\quad (k=0,1,2,\ldots,M), 
\label{Lambda_TM_Phi10=Phi10_TM}
\end{equation}
which corresponds to Lemma \ref{lemma:Lambda_T_Phi10=Phi10_T}.
\end{proof}

In the finite-carrier case, besides the invariants that are covered
by Theorem \ref{thm:1st_CQ}, we can consider yet another set of
invariants whenever $\lambda_M:=\left(\mathrm{asc}\circ
\Phi_{01}^{M-1}\right)(\bold{u})\neq 0$.
\begin{thm}
\label{thm:2nd_CQ} 
Fix an integer $M>1$. 
For $\bold{u}\in\mathcal{U}^{(M)}$, 
define 
$\bm{a}= \{a_1>\cdots>a_{\lambda_M}\}$ and 
$\bm{d}=\{d_1>\cdots>d_{\lambda_M}\}$ as 
\begin{equation}
\bm{a} = \mathrm{Asc}\left(\Phi_{01}^{M-1}(\bold{u})\right), \quad
\bm{d} = \mathrm{Des}\left(\Phi_{01}^{M-1}(\bold{u})\right), 
\end{equation}
i.e. $\Phi_{01}^{M-1}(\bold{u}) = \bold{u}(\bm{a};\bm{d})$. 
Denote by $\nu_j$ the length of the $j$th block (counted from the right, 
see Figure \ref{fig:def_nu}) of consecutive 1s in $\bold{u}(\bm{a};\bm{d})$,
\begin{equation}
\nu_j = d_j-a_j \quad (j=1,\ldots,\lambda_M). 
\label{nu=d-a}
\end{equation}
Then the composition $\nu=(\nu_j)_{j=1,2,\ldots,\lambda_M}$ is 
invariant under the action of $T_M$. 
\end{thm}
\begin{figure}[htbp]
\[
\begin{array}{l}
0\hspace{7mm}a_n\hspace{9mm}d_n\hspace{9mm}a_{n-1}\hspace{5mm}d_{n-1}
\hspace{14.5mm}a_2\hspace{10mm}d_2
\hspace{10mm}a_1\hspace{10mm}d_1
\\
0\;\cdots\; 0\; \underbrace{1\;\cdots\; 1}_{\nu_n}\; 
0\;\cdots\; 0\; \underbrace{1\;\cdots\; 1}_{\nu_{n-1}}\; 
\cdots\cdots\cdots\; 0\;\underbrace{1\;\cdots\; 1}_{\nu_2}\; 
0\;\cdots\; 0\;\underbrace{1\;\cdots\; 1}_{\nu_1}\; 0\;\cdots
\end{array}
\]
\caption{Definition of the composition $\nu$}
\label{fig:def_nu}
\end{figure}

\begin{proof}
Setting $k=M-1$ in \eqref{Lambda_TM_Phi10=Phi10_TM}, we have
\begin{equation}
\Phi_{01}^{M-1}\circ T_M 
= \Lambda^{M-1}\circ T_1\circ \Phi_{01}^{M-1}
= \Lambda^{M}\circ \Phi_{01}^{M-1}.
\end{equation}
This means that $\Phi_{01}^{M-1}(\bold{u})$ 
and $\left(\Phi_{01}^{M-1}\circ T_M \right)(\bold{u})$ coincide up to
a forward shift by $M$ steps (which is the maximum speed in the
system).
\end{proof}

Note that 
the existence of a $\nu_j>1$ implies that 
the corresponding soliton has a length greater than 
the maximum speed $M$ allowed by the time evolution $T_M$.

The data
$\bm{a}= \{a_1>\cdots>a_{\lambda_M}\}$ and
$\bm{d}=\{d_1>\cdots>d_{\lambda_M}\}$ in Theorem \ref{thm:2nd_CQ} 
satisfy the interlacing condition \eqref{InterlacingProperty}. 
Thus $\nu_j$ satisfies $0<\nu_j< a_{j-1}-a_j$ ($j=2,3,\ldots$).
For $\bold{u}\in\mathcal{U}^{(M)}$, we can define a modified version of 
\eqref{def:KKRmap}: 
\begin{equation}
KKR^{(M)}: \bold{u}\mapsto
\left\{
(\lambda_1,\ldots,\lambda_M),
\{J_1,\ldots,J_{M-1}\}, 
(\nu_1,\ldots,\nu_{\lambda_M}),
\{a_1,\ldots,a_{\lambda_M}\}
\right\}.
\label{modifiedRC}
\end{equation}
The data on the right-hand side of \eqref{modifiedRC} 
satisfy the conditions
\begin{equation}
\begin{aligned}
& \lambda_1\geq\cdots\geq\lambda_{M-1}\geq 0,\quad
J_{i,1}\geq\cdots\geq J_{i,m_i}\geq 0
\;\; (i=1,\ldots,M-1),\\
& 0<\nu_j< a_{j-1}-a_j \;\; (j=2,\ldots,\lambda_M), 
\end{aligned}
\end{equation}
with
$m_j=\lambda_j-\lambda_{j+1}$ ($j=1,\ldots,M-1$).
This can be represented graphically as in Figure \ref{fig:m-rigged_config}.
Gluing two diagrams in Figure \ref{fig:m-rigged_config}, 
with deleting the $M$th column of $\lambda$ (shaded in 
Figure \ref{fig:m-rigged_config}), 
gives the ``asymptotic soliton contents'' for the case with finite capacity 
(Figure \ref{fig:m-soliton_contents}).

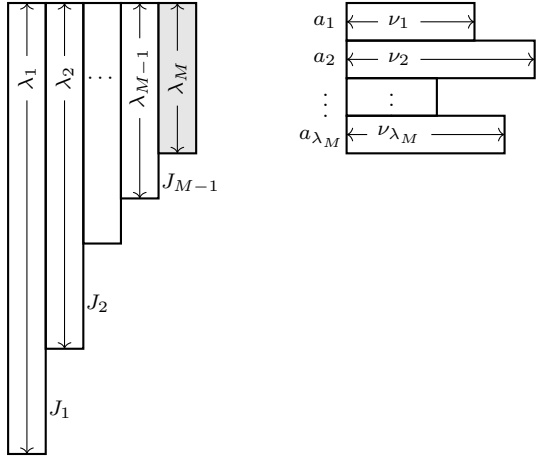
\begin{figure}[t!]
\begin{center}
\begin{tikzpicture}
\draw[thick] (0,0) rectangle (0.5,6);
  \node (lambda1) at (0.25,5) {\rotatebox{90}{\scriptsize $\lambda_1$}};
  \coordinate (B1) at (0.25,0);
  \coordinate (T1) at (0.25,6);
    \draw (0.4,0.6) node[anchor=west] {\scriptsize $J_1$};
\draw[thick] (0.5,1.4) rectangle (1,6);
  \node (lambda2) at (0.75,5) {\rotatebox{90}{\scriptsize $\lambda_2$}};
  \coordinate (B2) at (0.75,1.4);
  \coordinate (T2) at (0.75,6);
    \draw (0.9,2) node[anchor=west] {\scriptsize $J_2$};
\draw[thick] (1,2.8) rectangle (1.5,6);
  \node at (1.25,5) {\scriptsize $\cdots$};
\draw[thick] (1.5,3.4) rectangle (2,6);
  \node (lambdaM-1) at (1.75,5) {\rotatebox{90}{\scriptsize $\lambda_{M-1}$}};
    \draw (1.9,3.6) node[anchor=west] {\scriptsize $J_{M-1}$};
  \coordinate (BM-1) at (1.75,3.4);
  \coordinate (TM-1) at (1.75,6);
\draw[thick,fill=gray!20] (2,4) rectangle (2.5,6);
  \node (lambdaM) at (2.25,5) {\rotatebox{90}{\scriptsize $\lambda_M$}};
  \coordinate (BM) at (2.25,4);
  \coordinate (TM) at (2.25,6);
\draw[->] (lambda1) -- (T1);
\draw[->] (lambda1) -- (B1);
\draw[->] (lambda2) -- (T2);
\draw[->] (lambda2) -- (B2);
\draw[->] (lambdaM-1) -- (TM-1);
\draw[->] (lambdaM-1) -- (BM-1);
\draw[->] (lambdaM) -- (TM);
\draw[->] (lambdaM) -- (BM);
\draw[thick] (4.5,5.5) rectangle (6.2,6);
  \node (nu1) at (5.2,5.75) {\scriptsize $\nu_1$};
  \node[anchor=east] at (4.5,5.75) {\scriptsize $a_1$};
  \coordinate (L1) at (4.5,5.75);
  \coordinate (R1) at (6.2,5.75);
\draw[thick] (4.5,5) rectangle (7,5.5);
  \node (nu2) at (5.2,5.25) {\scriptsize $\nu_2$};
  \node[anchor=east] at (4.5,5.25) {\scriptsize $a_2$};
  \coordinate (L2) at (4.5,5.25);
  \coordinate (R2) at (7,5.25);
\draw[thick] (4.5,4.5) rectangle (5.7,5);
 \node at (5.1,4.75) {\scriptsize $\vdots$};
 \node[anchor=east] at (4.4,4.75) {\scriptsize $\vdots$};
\draw[thick] (4.5,4) rectangle (6.6,4.5);
  \node (nulambdaM) at (5.2,4.25) {\scriptsize $\nu_{\lambda_M}$};
  \node[anchor=east] at (4.6,4.2) {\scriptsize $a_{\lambda_M}$};
  \coordinate (LlambdaM) at (4.5,4.25);
  \coordinate (RlambdaM) at (6.6,4.25);
\draw[->] (nu1) -- (L1);
\draw[->] (nu1) -- (R1);
\draw[->] (nu2) -- (L2);
\draw[->] (nu2) -- (R2);
\draw[->] (nulambdaM) -- (LlambdaM);
\draw[->] (nulambdaM) -- (RlambdaM);
\end{tikzpicture}
\caption{Modified rigged configurations}
\label{fig:m-rigged_config}
\end{center}
\end{figure}

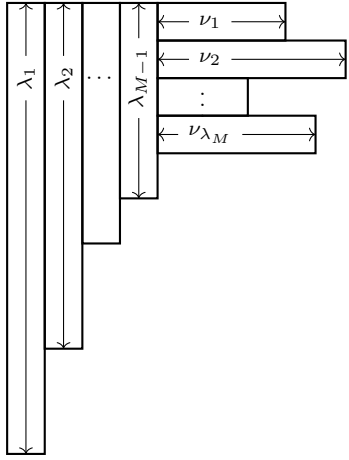
\begin{figure}[htbp]
\begin{center}
\begin{tikzpicture}
\draw[thick] (0,0) rectangle (0.5,6);
  \node (lambda1) at (0.25,5) {\rotatebox{90}{\scriptsize $\lambda_1$}};
  \coordinate (B1) at (0.25,0);
  \coordinate (T1) at (0.25,6);
\draw[thick] (0.5,1.4) rectangle (1,6);
  \node (lambda2) at (0.75,5) {\rotatebox{90}{\scriptsize $\lambda_2$}};
  \coordinate (B2) at (0.75,1.4);
  \coordinate (T2) at (0.75,6);
\draw[thick] (1,2.8) rectangle (1.5,6);
  \node at (1.25,5) {\scriptsize $\cdots$};
\draw[thick] (1.5,3.4) rectangle (2,6);
  \node (lambdaM') at (1.75,5) {\rotatebox{90}{\scriptsize $\lambda_{M-1}$}};
  \coordinate (BM') at (1.75,3.4);
  \coordinate (TM') at (1.75,6);
%
\draw[->] (lambda1) -- (T1);
\draw[->] (lambda1) -- (B1);
\draw[->] (lambda2) -- (T2);
\draw[->] (lambda2) -- (B2);
\draw[->] (lambdaM') -- (TM');
\draw[->] (lambdaM') -- (BM');
\draw[thick] (2,5.5) rectangle (3.7,6);
  \node (nu1) at (2.7,5.75) {\scriptsize $\nu_1$};
  \coordinate (L1) at (2,5.75);
  \coordinate (R1) at (3.7,5.75);
\draw[thick] (2,5) rectangle (4.5,5.5);
  \node (nu2) at (2.7,5.25) {\scriptsize $\nu_2$};
  \coordinate (L2) at (2,5.25);
  \coordinate (R2) at (4.5,5.25);
\draw[thick] (2,4.5) rectangle (3.2,5);
 \node at (2.6,4.75) {\scriptsize $\vdots$};
\draw[thick] (2,4) rectangle (4.1,4.5);
  \node (nulambdaM) at (2.7,4.25) {\scriptsize $\nu_{\lambda_M}$};
  \coordinate (LlambdaM) at (2,4.25);
  \coordinate (RlambdaM) at (4.1,4.25);
\draw[->] (nu1) -- (L1);
\draw[->] (nu1) -- (R1);
\draw[->] (nu2) -- (L2);
\draw[->] (nu2) -- (R2);
\draw[->] (nulambdaM) -- (LlambdaM);
\draw[->] (nulambdaM) -- (RlambdaM);
\end{tikzpicture}
\caption{Asymptotic soliton contents for $T_M$}
\label{fig:m-soliton_contents}
\end{center}
\end{figure}

\begin{example}
\label{eg:00001101100110100,T2}
$M=2$, 
$\bold{u}=00001101100110100\cdots\in\mathcal{U}^+$ 
(same initial condition as Example \ref{eg:00001101100110100,Tinfinity})
\begin{align*}
       \bold{u} &= 000011011001101000000000000\cdots\\
  T_2(\bold{u}) &= 000000101110010110000000000\cdots\\
T_2^2(\bold{u}) &= 000000010011101001100000000\cdots\\
T_2^3(\bold{u}) &= 000000001000110110011000000\cdots\\
T_2^4(\bold{u}) &= 000000000100001011100110000\cdots\\
T_2^5(\bold{u}) &= 000000000010000100111001100\cdots
\end{align*}
\[
\begin{array}{lll}
\lambda_1=\mathrm{asc}(\bold{u})=4, & 
\rho_{01}(\bold{u}) =\left\{3,6\right\}, &
\Phi_{01}(\bold{u}) = 000110100\cdots, \\
\lambda_2=\mathrm{asc}\left(\Phi_{01}(\bold{u})\right)=2, 
& \nu=\left\{1,2\right\}, & a_1=5,\quad a_2=2.
\end{array}
\]
\ytableausetup{boxsize=1em}
\begin{eqnarray*}
KKR^{(M)}(\bold{u})
& : &
\raisebox{-3mm}{$\Biggl\{$}\;
\begin{ytableau}
{} & *(lightgray) \\
{} & *(lightgray) \\
{} & \none[6]\\
{} & \none[3]
\end{ytableau}\;, \;
\begin{ytableau}
\none[5] & {}\\
\none[2] & {} & {}  
\end{ytableau}
\;\raisebox{-3mm}{$\Biggr\}$}
\quad \rightarrow  \quad \mbox{asymptotic soliton contents: \ }
\begin{ytableau}
{} & {} \\
{} & {} & {}\\
{} \\
{} 
\end{ytableau}
\end{eqnarray*}
\end{example}

\begin{thm}
\label{thm:soliton-contents:T_M}
Fix an integer $M>1$. 
For $\bold{u}\in\mathcal{U}^{(M)}$, 
define $\bm{a}$, $\bm{d}$, $\nu=(\nu_j)_{j=1,\ldots,\lambda_M}$ 
in the same way as in Theorem \ref{thm:2nd_CQ}. 
Denote by $\mu^{(j)}$ the conjugate partition of 
$(\lambda_1,\ldots,\lambda_{j})$ ($j=1,\ldots,M$). 
Then $\mu^{(M)}$ yields the speeds of the solitons. 
Moreover, the lengths of the solitons 
that are asymptotic with respect to $T_M$ 
are given by the formula (cf. Figure \ref{fig:m-soliton_contents}) 
\begin{equation}
\begin{cases}
\mu_j+\nu_j & (1\leq j\leq\lambda_M),\\
\mu_j & (\lambda_M<j\leq \lambda_1).
\end{cases}
\end{equation}
\end{thm}

A proof of this theorem will be given after the linearization 
property for this case has been explained.

In order to establish the linearization property for the time-evolution $T_M$, 
we first introduce the map $\kappa_M:\mathcal{U}\to\mathcal{U}$, 
\begin{equation}
\left(\kappa_M(\bold{u})\right)_n = 
\begin{cases}
0 & \mbox{if }(\bold{u})_n = \underline{1},\\
(\bold{u})_n & \mbox{otherwise}.
\end{cases}
\end{equation}
\begin{lemma}
\label{lemma:rho01_omega=rho01}
Assume $M\geq 2$. Then 
$\left(\rho_{01}\circ\kappa_M\right)(\bold{u})=\rho_{01}(\bold{u})$ 
for all $\bold{u}\in\mathcal{U}$.
\end{lemma}
\begin{proof}
It is clear that $\kappa_M$ does not affect the position of the $01$ pairs, 
and that, since the pattern ``$\underline{1}1$'' never occurs, no new
$01$ pairs can appear as a result of $\kappa_M$.
Because the pattern ``$0\underline{1}$'' also never occurs, 
the only pattern that needs to be considered is ``$01$'' (without underline). 

We consider the following patterns that include a $01$ pair in the middle,
and the action of $\kappa_M$ onto them: 
\begin{center}
(i) \ $0\arczeroone\underline{1}$, \qquad
(ii) \ $1\arczeroone\underline{1}$, \qquad
(iii) \ $\underline{1}\arczeroone 0$, \qquad
(iv) \ $\underline{1}\arczeroone 1$, \qquad
(v) \ $\underline{1}\arczeroone\underline{1}$.
\end{center}
However, the pattern (i) never occurs when $M\geq 2$, and 
the pattern (iv) never occurs because the rightmost $1$ 
should be $\underline{1}$. 
In the remaining patterns, the action of $\kappa_M$ has no influence on whether
the position of $\stackrel{\frown}{0\,1}$ is recorded or not.
\end{proof}

Note that because $\sigma\circ \rho_{01}$ is nothing but $\rho_{10}$ (as
in Lemma \ref{lemma:rho10=sigma_rho01}), we also find that
\begin{equation}
\rho_{10}\circ\kappa_M=\rho_{10}.
\end{equation}
Moreover, as for any $\bold{u}\in\mathcal{U}$ 
10-arc lines of $\kappa_M(\bold{u})$ coincide with 
01-arc lines of $T_M(\bold{u})$ it follows that
\begin{equation}
\label{rho10_omega=rho01_T}
\rho_{10}\circ\kappa_M=\rho_{01}\circ T_M,
\end{equation}
and hence we find that for $M\geq 2$, 
\begin{equation}
\label{rho10=rho01_TM}
\rho_{10}=\rho_{01}\circ T_M, 
\end{equation}
which is the finite capacity version of \eqref{lemma:rho10=rho01_T}.
Furthermore, as $T_1=\Lambda$, it is immediately clear from
\eqref{rho_Lambda=sigma_rho} that relation \eqref{rho10=rho01_TM} also
holds in the case $M=1$.

\begin{thm}[cf. \cite{KOSTY,Takagi}]
\label{thm:time-evolution_J_TM}
$\overline{J}_k 
:=\left(\rho_{01}\circ(\Phi_{01})^{k-1}\circ T_M\right)(\bold{u})
= \sigma^{k}\left(J_k\right)$
\ ($k=1,2,\ldots,M$)
\end{thm}
\begin{proof}
Using \eqref{rho_Lambda=sigma_rho}, \eqref{Lambda_TM_Phi10=Phi10_TM}, 
\eqref{rho10=rho01_TM} and Lemma \ref{lemma:rho10=sigma_rho01}, we have
\begin{align}
\overline{J}_k
&=\left(\rho_{01}\circ\Lambda^{k-1}\circ T_{M-k+1}\circ
 (\Phi_{01})^{k-1}\right)(\bold{u})
=\left(\sigma^{k}\circ\rho_{01}\circ
(\Phi_{01})^{k-1}\right)(\bold{u})
= \sigma^{k}\left(J_k\right) 
\end{align}
for $k=1,2,\ldots,M$. 
\end{proof}
\begin{rem}
In \cite{TakagiReview}, 
the linearization of $T_M$ (Theorem \ref{thm:time-evolution_J_TM}) 
has been proved based on the commutativity 
of $T_M$ and $T$ ($=T_{\infty}$).
\end{rem}

\begin{proof}[Proof of Theorem \ref{thm:soliton-contents:T_M}]
As in the proof of Theorem \ref{thm:SolitonContents_Tinfinity}, 
we know that 
there exists an integer $k$ such that the riggings at step $k$, 
$\left\{J_{i,j}(k)\right\}$, satisfy 
\begin{equation}
J_{i+1,m_{i+1}}-J_{i,1}>\mu_{\lambda_{i}}
\quad (i=1,\ldots,M-1).
\label{J-J>mu:i=1,...,M-1}
\end{equation}
As before, repeated application of the 
01-insertion map $\Psi_{01}$ results in 
a state that satisfies the condition 
\begin{equation}
J_{\tilde{\ell},1}\geq J_{\tilde{\ell},2}\geq\cdots 
\geq J_{\tilde{\ell},m_{\tilde{\ell}}}
> J_{\tilde{\ell}-1,1}\geq \cdots\cdots
\geq J_{2,m_2}
>J_{1,1}\geq\cdots \geq J_{1,m_1},
\end{equation}
where $\tilde{\ell}=\min\{\ell,M\}$ 
for $\ell$ defined as in Figure \ref{fig:rigged_config}. 
The $\mu_{\lambda_i}$ are therefore 
the asymptotic soliton speeds and lengths
for all $i=1,\ldots, M-1$, and also for $i=M$ if 
all $\nu_j$ are either $0$ or $1$.

In the remaining case, where at least one of the $\nu_j$ is 
greater than $1$, the pattern labeled by 
$\left(\nu_1,\ldots,\nu_{\lambda_M}\right)$ 
and $\left\{a_1,\ldots,a_{\lambda_M}\right\}$ is ``frozen'',
i.e. simply translates at speed $1$. 
In this case, 
under the condition \eqref{J-J>mu:i=1,...,M-1}, 
$M-1$ applications of the 
01-insertion map $\Psi_{01}$ result in 
solitons with lengths $\mu_j+\nu_j$, moving with speed $M$.
\end{proof}

\section{Combinatorial statistics and fermionic formula}
\label{sec:statistics}
In section \ref{sec:TS-BBA}, we have introduced the 
invariants $\lambda_k=\left(
\mathrm{asc}\circ\Phi_{01}^{k-1}\right)(\bold{u})$ and the riggings
$J_k=\left(\rho_{01}\circ\Phi_{01}^{k-1}\right)(\bold{u})$
for $\bold{u}\in\mathcal{U}^+$. 
Now, for a $\bold{u}\in\mathcal{U}$ that begins with $01$, the 
associated rigging $\rho_{01}(\bold{u})$ will contain the value $-1$. 
Taking $\bold{u}=010010000\cdots$ for example, we have
\begin{equation}
\bold{u}=\;\arczeroone 0\arczeroone 0\; 0\;0\;0\;
\cdots\quad\Rightarrow\quad 
\left\{\begin{aligned}
\Phi_{01}(\bold{u}) &= |\;0\;|\; 0\; 0\; 0\; 0\;\cdots, \\
\rho_{01}(\bold{u}) &= \{-1,0\}.
\end{aligned}\right.
\end{equation}
where the vertical lines correspond to ``0-solitons'', as before. 
To facilitate a combinatorial interpretation, 
we shall therefore consider a subset of $\mathcal{U}$ 
for which $\rho_{01}$ only takes non-negative values:
\begin{equation}
\tilde{\mathcal{U}} = \left\{
\bold{u}=(u_0,u_1,u_2,\ldots)\;|\; u_0=u_1=0, \;
u_j= 0\mbox{ or }1\;(j=2,3,\ldots)\right\}.
\end{equation}
We define $\tilde{\mathcal{U}}^{(n)}$ recursively by 
\begin{equation}
\tilde{\mathcal{U}}^{(1)} = \tilde{\mathcal{U}}, \quad 
\tilde{\mathcal{U}}^{(n)} = \Phi_{01}^{-1}\left(
\tilde{\mathcal{U}}^{(n-1)}\right) \cap \tilde{\mathcal{U}} 
\quad (n=2,3,\ldots),
\end{equation}
and $\tilde{\mathcal{U}}^+$ by
\begin{equation}
\tilde{\mathcal{U}}^+=\left\{
\bold{u}\in\tilde{\mathcal{U}}\,\Bigg|\,
\sum_{j=1}^{k}\chi\left((\bold{u})_j=0\right)
\geq \sum_{j=1}^{k}\chi\left((\bold{u})_j=1\right)
\mbox{ for }k=1,2,\ldots
\right\}.\end{equation}
It follows that
\begin{equation}
\tilde{\mathcal{U}}=\tilde{\mathcal{U}}^{(1)}\supset
\tilde{\mathcal{U}}^{(2)}\supset\cdots\supset
\tilde{\mathcal{U}}^{(n)}\supset\cdots\supset
\tilde{\mathcal{U}}^{+}.
\end{equation}

For positive integers $N$, $k$, $m$ such that $N\geq 2k$, 
define $\mathcal{P}^{(m)}(N,k)$ and $\mathcal{P}^{+}(N,k)$ as
\begin{align}
\mathcal{P}^{(m)}(N,k) &= \left\{
\bold{u}\in\tilde{\mathcal{U}}^{(m)}\,\Big|\,
\mathcal{N}(\bold{u})=k,\; 
(\bold{u})_j=0 \mbox{ for all } j>N
\right\}, 
\label{Pm(N,k)}\\
\mathcal{P}^{+}(N,k) &= \left\{
\bold{u}\in\tilde{\mathcal{U}}^{+}\,\Big|\,
\mathcal{N}(\bold{u})=k,\; 
(\bold{u})_j=0 \mbox{ for all } j>N
\right\}.
\label{P+(N,k)}
\end{align}
Since in the states that belong to the sets \eqref{Pm(N,k)} and \eqref{P+(N,k)}
all $u_j$ are zero beyond $j=N$, we can regard the number $N+1$ as
representing the total number of boxes to consider in our combinatorial
problem, and $k$ as the number of balls that go into those boxes. 
Moreover, since $\lambda_i$ is equal to the number of 01-pairs that are eliminated 
in the $i$th application of $\Phi_{01}$, it is clear that the $i$th \textit{vacancy number} $p_i$ (cf. \eqref{def:vacancy})
\begin{equation}
p_i= N-2\sum_{k=1}^{i}\lambda_k \quad (i=1,\ldots,\ell), 
\end{equation}
can be interpreted as the number of remaining boxes after $i$ 01-eliminations.
Given $\bold{u}\in\mathcal{P}^{(n)}(N,k)$, 
define a partition $\lambda=\left(\lambda_1,\ldots,\lambda_\ell\right)$
as in \eqref{def:lambda_i} and, as before, take $m_i$ ($i=1,2,\ldots,\ell$) to be
\begin{equation}
m_i=
\begin{cases}
\lambda_i-\lambda_{i+1} & (i=1,\ldots,\ell-1),\\
\lambda_\ell & (i=\ell).
\end{cases}
\label{def:m:=lambda-lambda}
\end{equation}
The riggings $J_i=\{J_{i,1},\ldots,J_{i,m_i}\}$ of course 
satisfy
\begin{equation}
p_i\geq J_{i,1}\geq\cdots\geq J_{i,m_i} \geq0.
\label{vacancy_condition}
\end{equation}

Given a partition $\lambda=\left(\lambda_1,\ldots,\lambda_\ell\right)$
of length $\ell$, 
we denote by $\mathcal{P}^+\left(N;\lambda\right)$ the following 
finite subset of 
$\mathcal{P}^+\left(N,{\textstyle \sum_i}\lambda_i\right)$:
\begin{equation}
\mathcal{P}^+\left(N;\lambda\right)=
\left\{
\bold{u}\in\mathcal{P}^+\big(N,{\textstyle \sum_i}\lambda_i\big)\:\Big|\:
\mathrm{asc}\left(\Phi_{01}^{i-1}(\bold{u})\right)=\lambda_i 
\; (i=1,2,\ldots,\ell),\;
\Phi_{01}^\ell(\bold{u})=\bold{0}
\right\}, 
\end{equation}
and by $\mathrm{Rig}\left(N;\lambda\right)$ the set of 
possible riggings that correspond to a 
$\bold{u}\in\mathcal{P}^+\left(N;\lambda\right)$:
\begin{equation}
\mathrm{Rig}\left(N;\lambda\right)
=\left\{J_{i,k}\; (i=1,\ldots,\ell,\; k=1,\ldots,m_i) 
\mbox{ that satisfy the condition \eqref{vacancy_condition}}\right\}.
\end{equation}
Then the map 
\begin{equation}
\begin{array}{rcl}
\mathcal{P}^+\left(N;\lambda\right)
 & \to & \mathrm{Rig}\left(N;\lambda\right)\\
\bold{u} & \mapsto & 
\{J_{i,1}\geq\cdots\geq J_{i,m_i}\}
=\rho_{01}\left(\Phi_{01}^{i-1}(\bold{u})\right)
\quad (i=1,2,\ldots,\ell)
\end{array}
\end{equation}
is a bijection (KKR bijection) since the partition $\lambda$
can be reconstructed from the riggings:
\begin{equation}
\lambda_{\ell-i} = m_\ell + m_{\ell-1}+\cdots + m_{\ell-i}
\quad (i=0,\ldots,\ell-1).
\end{equation}
Note that
\begin{equation}
\mathcal{P}^+(N,k) =
\bigsqcup_{\lambda\,\vdash\,k}\mathcal{P}^+(N;\lambda), 
\end{equation}
where $\lambda \vdash k$ means that $\lambda$ is a partition of $k$.

\begin{example}[$N=6$, $k=3$]
\begin{align*}
\mathcal{P}^+(N=6,k=3)
&= \mathcal{P}^+\left(6;\left(1,1,1\right)\right)
\sqcup\mathcal{P}^+\left(6;\left(2,1\right)\right)
\sqcup\mathcal{P}^+\left(6;\left(3\right)\right),
\\
\mathcal{P}^+\left(6;\left(1,1,1\right)\right)
&=\left\{0000111\right\},
\\
\mathcal{P}^+\left(6;\left(2,1\right)\right)
&=\left\{0010011,\; 0001011,\; 0001101\right\},
\\
\mathcal{P}^+\left(6;\left(3\right)\right)
&=\left\{0010101\right\}.
\end{align*}
\end{example}

Before considering combinatorial statistics on
$\mathcal{P}^{+}(N;\lambda)$, 
we must prepare some notation. 
Given a sequence $\bold{u}\in\mathcal{U}$, we define 
$\mathrm{maj}(\bold{u})$ and $\mathrm{comaj}(\bold{u})$ as 
\begin{equation}
\mathrm{maj}(\bold{u}) = \sum_{i\in\mathrm{Des}(\bold{u})}i,\qquad
\mathrm{comaj}(\bold{u}) = \sum_{i\in\mathrm{Asc}(\bold{u})}i,
\end{equation}
which are known as the major index and the comajor index, 
respectively \cite{MacMahon,Mansour,Krattenthaler}.
We remark that the comajor index $\mathrm{comaj}(\bold{u})$ is 
equivalent to the energy of $\bold{u}$ discussed in 
\cite{SchillingReview,Takagi,TakagiReview}.

\begin{lemma}
$\mathrm{maj}(\bold{u})=\mathrm{comaj}(\bold{u})
+\mathcal{N}(\bold{u})$
for any $\bold{u}\in\mathcal{U}$.
\end{lemma}
\begin{proof}
Consider the quantity $\nu_j = d_j-a_j$ 
($j=1,\ldots,\lambda_1$) defined as \eqref{nu=d-a} 
in Theorem \ref{thm:2nd_CQ}, 
which represents the length of the $j$th block (from the right) 
of consecutive 1s. The desired result follows from 
$\sum_{j}\nu_j = \mathcal{N}(\bold{u})$.
\end{proof}

\begin{thm}
\label{thm:comaj_recursion}
Given $\bold{u}\in\mathcal{U}$, we have
\begin{align}
\mathrm{comaj}\left(\bold{u}\right)
&= \mathrm{comaj}\left(\Phi_{01}(\bold{u})\right)
+\mathrm{asc}(\bold{u})^2 + \sum_{j\in\rho_{01}(\bold{u})}j, 
\label{comaj_recursion}\\
\mathrm{maj}\left(\bold{u}\right)
&= \mathrm{maj}\left(\Phi_{10}(\bold{u})\right)
+\mathrm{des}(\bold{u})^2 + \sum_{j\in\rho_{10}(\bold{u})}j.
\label{maj_recursion}
\end{align}
\end{thm}
\begin{proof}
Set $N=\mathrm{asc}(\bold{u})$ and 
$\bm{a}=\{a_1>\cdots >a_N\}=\mathrm{Asc}(\bold{u})$. 
Then $\mathrm{comaj}(\bold{u})=\sum_{i=1}^N a_i$.

Let $I'_1,\ldots,I'_N$ be the sets of integers constructed from 
$\bold{u}$ by the recursion relations \eqref{recursion:u'},
\eqref{recusionRel_I'j} with the initial condition $I'_0=\emptyset$.
{}From the definition \eqref{def:f_c}, it follows that
\begin{equation}
f_{\bm{a}}(a_i)=a_i - 2N+2i-1, 
\end{equation}
and
\begin{equation}
f_{\bm{a}}(I'_N) = \rho_{01}(\bold{u}),\quad
f_{\bm{a}}(\bm{a}\backslash I'_N) = \mathrm{Asc}(\Phi_{01}(\bold{u})).
\end{equation}
Thus we have
\begin{align}
\mathrm{comaj}(\bold{u})&=\sum_{i=1}^N a_i
=\sum_{i=1}^N\left\{
f_{\bm{a}}(a_i)+2N+1-2i\right\}
\nonumber\\
&= N^2 +\sum_{i=1}^ Nf_{\bm{a}}(a_i)
= N^2+\mathrm{comaj}(\Phi_{01}(\bold{u})) + 
\sum_{j\in\rho_{01}(\bold{u})}j. 
\end{align}
Thus we have obtained \eqref{comaj_recursion}. 
The remaining relation \eqref{maj_recursion} can be proved 
along the same lines. 
\end{proof}

We consider the generating function of the comajor index on 
the finite subset $\mathcal{P}^+(N;\lambda)$:
\begin{equation}
Z\left(N;\lambda\right)
= \sum_{\bold{u}\in\mathcal{P}^+(N;\lambda)}q^{\mathrm{comaj}(\bold{u})}.
\label{Z(N,lambda)}
\end{equation}
In what follows, we shall show that the generating function 
\eqref{Z(N,lambda)} can be expressed in terms of 
$q$-binomial coefficients, 
\begin{equation}
\begin{bmatrix}
m+n\\ m
\end{bmatrix}_q = 
\frac{[m+n]_q!}{[m]_q![n]_q!}, 
\quad
[n]_q! = \prod_{k=1}^n [k]_q, \quad
[k]_q = \frac{1-q^k}{1-q}. 
\end{equation}
We first prepare a lemma. 
\begin{lemma}[\cite{MacMahon}]
\label{lemma:MacMahon}
\begin{equation}
\sum_{p\geq j_1\geq j_2\geq\cdots\geq j_m\geq 0}
q^{\sum_{k=1}^mj_k} 
= \begin{bmatrix}p+m\\ m\end{bmatrix}_q
\label{eq:MacMahon}
\end{equation}
\end{lemma}
This can be proved by showing the both sides of 
\eqref{eq:MacMahon} satisfy the same recursion
\begin{equation}
\begin{bmatrix}p+m+1\\ m+1\end{bmatrix}_q
=\begin{bmatrix}p+m\\ m\end{bmatrix}_q+
q^{m+1}\begin{bmatrix}p+m\\ m+1\end{bmatrix}_q, 
\end{equation}
and the boundary condition
\begin{equation}
\begin{bmatrix}n\\ 0\end{bmatrix}_q
=\begin{bmatrix}n\\ n\end{bmatrix}_q=1, 
\quad 
\begin{bmatrix}n\\ 1\end{bmatrix}_q = \frac{1-q^n}{1-q}.
\end{equation}
\begin{thm}[Fermionic formula of $A_1^{(1)}$-type
\cite{HKOTT,OkadoRevrew,Schilling:X=M}]
\label{thm:Z(N,lambda)}
$\displaystyle Z(N,\lambda) = \prod_{i=1}^\ell q^{\lambda_i^2}
\begin{bmatrix}p_i+m_i\\ m_i\end{bmatrix}_q$, 
where $m_i$ and $p_i$ are as defined in formulas 
\eqref{def:vacancy} and \eqref{def:m:=lambda-lambda}, 
 for the partition $\lambda=(\lambda_1,\ldots,\lambda_\ell)$.
\end{thm}
\begin{proof}
Using Theorem \ref{thm:comaj_recursion} recursively, we have
\begin{equation}
\mathrm{comaj}(\bold{u}) = \sum_{i=1}^\ell \left(
\lambda_i^2 + \sum_{k=1}^{m_i}J_{i,k}\right)
\end{equation}
for $\bold{u}\in\mathcal{P}^+
\left(N,\left(\lambda_1,\ldots,\lambda_\ell\right)\right)$.
We therefore obtain
\begin{align}
Z(N,\lambda) &= 
\sum_{\bold{u}\in\mathcal{P}^+(N;\lambda)}q^{\mathrm{comaj}(\bold{u})}
=\sum_{J_1}\sum_{J_2}\cdots\sum_{J_\ell}
q^{\sum_{i=1}^\ell \left(
\lambda_i^2 + \sum_{k=1}^{m_i}J_{i,k}\right)}
\nonumber\\
&= q^{\lambda_1^2+\lambda_2^2+\cdots+\lambda_\ell^2}
\sum_{J_1}q^{\sum_{k_1=1}^{m_1}J_{1,k_1}}
\sum_{J_2}q^{\sum_{k_2=1}^{m_2}J_{2,k_2}}
\cdots
\sum_{J_\ell}q^{\sum_{k_\ell=1}^{m_\ell}J_{\ell,k_\ell}}.
\label{Z(N,lambda)=...}
\end{align}
Applying \eqref{eq:MacMahon} to \eqref{Z(N,lambda)=...}, 
we have the desired result since 
the $J_i=\{J_{i,1},\ldots, J_{i,m_i}\}$ satisfy condition \eqref{vacancy_condition}. 
\end{proof}
\begin{rem}
Theorem \ref{thm:Z(N,lambda)} is related to a natural $q$-analogue of 
the Catalan numbers \cite{FH,Kirillov,Reynolds}, defined by
\begin{equation}
C_n(q) = \frac{1}{\left[n+1\right]_q}
\begin{bmatrix}2n\\ n\end{bmatrix}_q, 
\label{def:q-Catalan}
\end{equation}
which is connected to a comajor counting as 
\begin{equation}
C_n(q)
=q^{-n}\sum_{\bold{u}\in\mathcal{P}^+(2n,n)}q^{\mathrm{comaj}(\bold{u})}
=q^{-n}\sum_{|\lambda|=n}\prod_{i=1}^{\ell(\lambda)} q^{\lambda_i^2}
\begin{bmatrix}p_i+m_i\\ m_i\end{bmatrix}_q.
\label{Z=qCatalan}
\end{equation}
\end{rem}

\begin{example}[$N=10$, $k=5$, $\lambda=\left(2,2,1\right)$]
\label{example:N=10,k=5,lambda=(2,2,1)}
\[
\mathcal{P}^+\left(10;\left(2,2,1\right)\right)
=\left\{00011000111,\;
00001100111,\;
00001110011\right\}.
\]
As can be seen from Table \ref{table:N=10,k=5,lambda=(2,2,1)}, 
the identity of Theorem \ref{thm:Z(N,lambda)} in this case gives
\[
Z\left(10,\left(2,2,1\right)\right)=
q^{2+7}+q^{3+7}+q^{3+8} = 
q^{2^2}\begin{bmatrix}6+0\\ 0\end{bmatrix}_q
\cdot
q^{2^2}\begin{bmatrix}2+1\\ 1\end{bmatrix}_q
\cdot
q^{1^2}\begin{bmatrix}0+1\\ 1\end{bmatrix}_q.
\]
\begin{table}[htbp]
\ytableausetup{boxsize=1em}
\begin{center}
\begin{tabular}{c|c|c|c}
$\bold{u}=u_0u_1u_2\cdots u_{10}$ & 
$\mathrm{Asc}(\bold{u})$ & $\left\{J_1,J_2,J_3\right\}$
& {\small rigged configuration}\\ \hline &&& \\[-2mm]
$00011000111$ & $\left\{7,2\right\}$ & 
$\left\{\emptyset,\{0\},\{0\}\right\}$
& $\begin{ytableau}
{} & {} & {} & \none[0]\\
{} & {} & \none[0]
\end{ytableau}$
\\[6mm]
$00001100111$ & $\left\{7,3\right\}$ & 
$\left\{\emptyset,\{1\},\{0\}\right\}$
& $\begin{ytableau}
{} & {} & {} & \none[0]\\
{} & {} & \none[1]
\end{ytableau}$
\\[6mm]
$00001110011$ & $\left\{8,3\right\}$ & 
$\left\{\emptyset,\{2\},\{0\}\right\}$
& $\begin{ytableau}
{} & {} & {} & \none[0]\\
{} & {} & \none[2]
\end{ytableau}$
\end{tabular}
\caption{$\mathcal{P}^+\left(10;\left(2,2,1\right)\right)$}
\label{table:N=10,k=5,lambda=(2,2,1)}
\end{center}
\end{table}
\end{example}

We now consider partition functions associated with 
the modified version of the KKR map \eqref{modifiedRC}.
Given a partition 
$\lambda=\left(\lambda_1\geq\cdots\geq\lambda_M\right)$ of 
length $M$, a composition 
$\nu=\left(\nu_1,\ldots,\nu_{\lambda_M}\right)$ 
of length $\lambda_M$, 
and an integer $N\geq
\lambda_1+\cdots+\lambda_{M-1}+\nu_1+\cdots+\nu_{\lambda_M}$, 
we set 
\begin{align}
\mathcal{P}^{(M)}(N;\lambda,\nu) =
\left\{\bold{u}\in\mathcal{P}^{(M)}(N,|\lambda|+|\nu|)\,\right|\,
& \mathrm{asc}\left(\Phi_{01}^{i-1}(\bold{u})\right)=\lambda_i
\; (i=1,\ldots,M), 
\nonumber\\
&\left. 
\nu_j = d^{(M-1)}_j-a^{(M-1)}_j
\; (j=1,\ldots,m_{M})
\right\},
\end{align}
where $\left\{d^{(M-1)}_j\right\}:= 
\mathrm{Des}\left(\Phi_{01}^{M-1}(\bold{u})\right)$ and 
$\left\{a^{(M-1)}_j\right\}:= 
\mathrm{Asc}\left(\Phi_{01}^{M-1}(\bold{u})\right)$. 

\begin{example}[$M=2$, $N=10$, $\lambda=\left(2,2\right)$, 
$\nu=\left(2,1\right)$]
\label{example:M=2,N=10,lambda=(2,2),mu=(2,1)}
\[
\mathcal{P}^{(2)}\left(10;\left(2,2\right),\left(2,1\right)\right) =
\left\{
00011001110,\;
00011000111,\;
00001100111
\right\}.
\]
\ytableausetup{boxsize=1em}
\begin{table}[htbp]
\begin{center}
\begin{tabular}{c|c|c}
$\bold{u}=u_0u_1u_2\cdots u_{10}$ & 
$\mathrm{Asc}(\bold{u})$ 
& {\small modified rigged configuration}\\ \hline && \\[-3mm]
$00011001110$ & $\left\{6,2\right\}$ 
& $\biggl\{\;
\raisebox{1mm}{$
\begin{ytableau}
{} & *(lightgray) \\
{} & *(lightgray) 
\end{ytableau}\;, \;
\begin{ytableau}
\none[3] & {} & {}\\
\none[1] & {}
\end{ytableau}$}
\;\biggr\}$
\\[3mm]
$00011000111$ & $\left\{7,2\right\}$
& $\biggl\{\;
\raisebox{1mm}{$
\begin{ytableau}
{} & *(lightgray) \\
{} & *(lightgray) 
\end{ytableau}\;, \;
\begin{ytableau}
\none[4] & {} & {}\\
\none[1] & {}
\end{ytableau}$}
\;\biggr\}$
\\[3mm]
$00001100111$ & $\left\{7,3\right\}$
& $\biggl\{\;
\raisebox{1mm}{$
\begin{ytableau}
{} & *(lightgray) \\
{} & *(lightgray) 
\end{ytableau}\;, \;
\begin{ytableau}
\none[4] & {} & {}\\
\none[2] & {}
\end{ytableau}$}
\;\biggr\}$
\end{tabular}
\end{center}
\caption{$\mathcal{P}^{(2)}\left(10;\left(2,2\right),
\left(2,1\right)\right)$}
\label{table:P(2)(10;(2,2))}
\end{table}
\end{example}
\begin{rem}
The associated soliton contents 
(asymptotic length of solitons) 
for Example \ref{example:N=10,k=5,lambda=(2,2,1)} 
(as given in Table \ref{table:N=10,k=5,lambda=(2,2,1)})
and those for 
Example \ref{example:M=2,N=10,lambda=(2,2),mu=(2,1)} 
(Table \ref{table:P(2)(10;(2,2))}) 
are the same, given by 
$\begin{ytableau}
{} & {} & {}\\
{} & {} 
\end{ytableau}$.
\end{rem}
\begin{example}[$M=2$, $N=10$, $\lambda=\left(2,2\right)$, 
$\nu=\left(1,2\right)$]
\label{example:M=2,N=10,lambda=(2,2),mu=(1,2)}
\[
\mathcal{P}^{(2)}\left(10;\left(2,2\right),\left(1,2\right)\right) =
\left\{
00011100110,\;
00011100011,\;
00001110011
\right\}.
\]
\ytableausetup{boxsize=1em}
\begin{table}[htbp]
\begin{center}
\begin{tabular}{c|c|c}
$\bold{u}=u_0u_1u_2\cdots u_{10}$ & 
$\mathrm{Asc}(\bold{u})$ 
& {\small modified rigged configuration}\\ \hline && \\[-3mm]
$00011100110$ & $\left\{7,2\right\}$ 
& $\biggl\{\;
\raisebox{1mm}{$
\begin{ytableau}
{} & *(lightgray) \\
{} & *(lightgray) 
\end{ytableau}\;, \;
\begin{ytableau}
\none[4] & {} \\
\none[1] & {} & {}
\end{ytableau}$}
\;\biggr\}$
\\[3mm]
$00011100011$ & $\left\{8,2\right\}$
& $\biggl\{\;
\raisebox{1mm}{$
\begin{ytableau}
{} & *(lightgray) \\
{} & *(lightgray) 
\end{ytableau}\;, \;
\begin{ytableau}
\none[5] & {}\\
\none[1] & {} & {}
\end{ytableau}$}
\;\biggr\}$
\\[3mm]
$00001110011$ & $\left\{8,3\right\}$
& $\biggl\{\;
\raisebox{1mm}{$
\begin{ytableau}
{} & *(lightgray) \\
{} & *(lightgray) 
\end{ytableau}\;, \;
\begin{ytableau}
\none[5] & {}\\
\none[2] & {} & {}
\end{ytableau}$}
\;\biggr\}$
\end{tabular}
\end{center}
\caption{$\mathcal{P}^{(2)}\left(10;\left(2,2\right),
\left(1,2\right)\right)$}
\end{table}
\end{example}
For a partition $\lambda=\left(\lambda_1,\ldots,\lambda_M\right)$ 
and a composition $\nu=\left(\nu_1,\ldots,\nu_{\lambda_M}\right)$, 
we define the space of riggings for the finite carrier case as 
\begin{align}
&\mathrm{Rig}^{(M)}(N;\lambda,\nu) =
\nonumber\\
&\left\{
\begin{array}{l}
J_1 = \left\{J_{1,1},\ldots,J_{1,m_1}\right\}\in\mathbb{Z}^{m_1},\\
\quad\cdots\cdots\cdots\\
J_{\ell-1} = \left\{
J_{M-1,1},\ldots,J_{M-1,m_{\ell-1}}\right\}\in\mathbb{Z}^{m_{M-1}},\\
a_1,\ldots,a_{\lambda_M}\in\mathbb{Z}
\end{array}\right|
\left.
\begin{array}{l}
p_1\geq J_{1,1}\geq\ldots\geq J_{1,m_1}\geq 0,\\
\quad\cdots\cdots\cdots\\
p_{M-1}\geq J_{M-1,1}\geq\ldots\geq J_{M-1,m_{M-1}}\geq 0,\\
a_1 + \nu_1\leq N, \quad
a_j +\nu_j < a_{j-1} \; (j=2,\ldots,\lambda_M)
\end{array}
\right\}.
\end{align}
The map \eqref{modifiedRC} is a bijection from 
$\mathcal{P}^{(M)}(N;\lambda,\nu)$ to
$\mathrm{Rig}^{(M)}(N;\lambda,\nu)$. 
In other words, the data 
$\left\{
\{J_1,\ldots,J_{M-1}\},\right.$ 
$\left.\{a_1,\ldots,a_{\lambda_M}\}
\right\}\in\mathrm{Rig}^{(M)}(N;\lambda,\nu)$ 
uniquely parametrises an element 
$\bold{u}\in\mathcal{P}^{(M)}(N;\lambda,\nu)$. 

Consider the comajor counting associated with 
$\mathcal{P}^{(M)}(N;\lambda,\nu)$: 
\begin{equation}
Z^{(M)}(N;\lambda,\nu) = 
\sum_{\bold{u}\in\mathcal{P}^{(M)}(N;\lambda,\nu)}q^{\mathrm{comaj}(\bold{u})}.
\end{equation}

\begin{thm}[Fermionic formula for a BBS with a carrier of finite
 capacity $M$]
\label{thm:FermionicFormula_finiteM}
\begin{equation}
Z^{(M)}(N;\lambda,\nu) = 
\left(\prod_{j=1}^{M-1}q^{\lambda_j^2}
\begin{bmatrix}p_j+m_j\\ m_j\end{bmatrix}_q\right)
\left(q^{\lambda_M(\lambda_M+1)/2}
\prod_{k=1}^{\lambda_M} q^{(k-1)\nu_{k}}\right)
\begin{bmatrix}p_{M-1}-|\nu|\\ \lambda_M\end{bmatrix}_q.
\label{eq:Z(N,lambda,nu)}
\end{equation}
\end{thm}
\begin{example}
In the cases of 
Example \ref{example:M=2,N=10,lambda=(2,2),mu=(2,1)} 
and Example \ref{example:M=2,N=10,lambda=(2,2),mu=(1,2)}, 
the identity \eqref{eq:Z(N,lambda,nu)} gives
\begin{align*}
Z^{(2)}\left(10;\left(2,2\right),\left(2,1\right)\right)
&=q^{2+6}+q^{2+7}+q^{3+7} = 
q^{2^2}\begin{bmatrix}6+0\\ 0\end{bmatrix}_q
\cdot \left(q^{3}q^{1}\right)
\cdot\begin{bmatrix}6-3\\ 2\end{bmatrix}_q,
\\
Z^{(2)}\left(10;\left(2,2\right),\left(1,2\right)\right)
&=q^{2+7}+q^{2+8}+q^{3+8} = 
q^{2^2}\begin{bmatrix}6+0\\ 0\end{bmatrix}_q
\cdot \left(q^{3}q^{2}\right)
\cdot\begin{bmatrix}6-3\\ 2\end{bmatrix}_q.
\end{align*}
\end{example}

Theorem \ref{thm:FermionicFormula_finiteM} 
follows directly form 
Theorem \ref{thm:comaj_recursion} and 
\begin{lemma}
Given a composition $\nu=\left(\nu_1,\ldots,\nu_n\right)$ and an integer 
$N'\geq \nu_1+\cdots +\nu_n+n$, we define
\begin{equation}
\mathcal{P}'(N';\nu) = \left\{
\bold{u}\in\tilde{\mathcal{U}}\,\Big|\, (\bold{u})_j=0\mbox{ if }j>N', \;
\mbox{length of the $j$th soliton}=\nu_j \;(j=1,\ldots,n)
\right\}.
\end{equation}
Then we have
\begin{equation}
\sum_{\bold{u}\in\mathcal{P}'(N';\nu)}q^{\mathrm{comaj}(\bold{u})}
=q^{n(n+1)/2}\left(\prod_{k=1}^{n-1}q^{k\nu_{k+1}}\right)
\begin{bmatrix}
N'-(\nu_1+\cdots +\nu_n) \\ n
\end{bmatrix}_q.
\end{equation}
\end{lemma}
\begin{proof}
Given $\bold{u}\in\mathcal{P}'
\left(N';\left(\nu_1,\nu_2,\ldots,\nu_n\right)\right)$, we define
$\{a_1>\cdots>a_n\}=\mathrm{Asc}(\bold{u})$, 
$\{d_1>\cdots>d_n\}=\mathrm{Des}(\bold{u})$, 
where $a_j$ and $d_j$ satisfy the interlacing condition
(see Figure \ref{fig:def_nu})
\begin{equation}
1\leq a_n<d_n<a_{n-1}<d_{n-1}<\cdots<a_1<d_1\leq N'.
\end{equation}

Define $l_k$ ($k=1,\ldots,n$) as
\begin{equation}
l_k = a_k -(n+1-k) - \sum_{j=k+1}^n \nu_j. 
\end{equation}
These satisfy 
\begin{equation}
0\leq l_n\leq l_{n-1}\leq \cdots\leq l_2
\leq l_1\leq N'-(\nu_1+\cdots +\nu_n)-n
=: N''.
\end{equation}
It follows that
\begin{equation}
\mathrm{comaj}(\bold{u}) 
= a_1+\cdots +a_n 
= \sum_{k=1}^nl_k + \sum_{k=1}^{n}(k-1)\nu_{k} + \frac{n(n+1)}{2}, 
\end{equation}
and hence we have
\begin{align}
\sum_{\bold{u}\in\mathcal{P}'(N';\nu)}q^{\mathrm{comaj}(\bold{u})}
&=\sum_{0\leq l_n\leq\cdots\leq l_1\leq N''}
 q^{\sum_{j}l_j + \sum_{k=1}^{n}(k-1)\nu_{k} + n(n+1)/2}
\nonumber\\
&=q^{n(n+1)/2}\left(\prod_{k=1}^{n-1}q^{k\nu_{k+1}}\right)
\begin{bmatrix}N'-\sum_{j}\nu_j\\ n\end{bmatrix}_q,
\end{align}
where we have used Lemma \ref{lemma:MacMahon}.
\end{proof}

\section*{Acknowledgments}
The authors are grateful to Atsuo Kuniba, Masato Okado, and 
Taichiro Takagi for the clear explanations they gave of 
their results and for helpful comments concerning our work. 
We also thank Claire Gilson and Christian Korff 
for discussions and useful comments.
Finally, we would like to express our gratitude to the anonymous referee for
valuable comments that greatly strengthened our paper. 
This work was supported by JSPS KAKENHI Grant
Numbers 15K04893, 16K05184, 16K13761.

\end{document}